\documentclass[11pt]{article}
\usepackage{fullpage}
\usepackage{amsmath,amsthm,amssymb}
\usepackage{xcolor,colortbl}
\usepackage{natbib}
\usepackage[ruled,vlined]{algorithm2e}
\usepackage{tikz}
\usepackage{subcaption}
\usepackage{algpseudocode}
\usepackage{hyperref}
\usepackage[nameinlink,capitalise]{cleveref}

\makeatletter
\def\@fnsymbol#1{\ensuremath{\ifcase#1\or \dagger\or \ddagger\or
   \mathsection\or \mathparagraph\or \|\or **\or \dagger\dagger
   \or \ddagger\ddagger \else\@ctrerr\fi}}
\makeatother
\newcommand*\samethanks[1][\value{footnote}]{\footnotemark[#1]}

\newtheorem{theorem}{Theorem}[section]

\newtheorem{lemma}[theorem]{Lemma}
\newtheorem{proposition}[theorem]{Proposition}

\theoremstyle{definition}
\newtheorem{example}[theorem]{Example}
\newtheorem{definition}[theorem]{Definition}

\crefname{program}{Program}{Programs}

\newcommand{\rev}{\mathsf{Rev}}

\renewcommand{\S}{\mathcal{S}}
\newcommand{\p}{\mathbf{p}}

\newcommand{\nt}{N_\mathrm{T}}
\newcommand{\I}{\mathcal{I}}

\newcommand{\nxt}{\mathrm{next}}

\renewcommand{\L}{\mathcal{L}}

\DeclareMathOperator{\E}{\mathbb{E}}

\definecolor{myblue}{rgb}{0.15, 0.1, 0.95}
\definecolor{mygreen}{rgb}{0.15, 0.65, 0.25}
\definecolor{myred}{rgb}{0.75, 0.25, 0.15}

\hypersetup{
    colorlinks=true,
    linkcolor=myred,
    citecolor=myblue,
    urlcolor=mygreen,
}

\title{Optimal Pricing Schemes for an Impatient Buyer}

\author{Yuan Deng\thanks{Google Research. Email: \texttt{\{dengyuan,maojm,balusivan\}@google.com}.} \and Jieming Mao\samethanks[1] \and Balasubramanian Sivan\samethanks[1] \and Kangning Wang\thanks{Stanford University. Email: \texttt{knwang@stanford.edu}. This work was partially done while the author was at Duke University and was an intern at Google Research.}}

\date{}

\begin{document}
\maketitle

\begin{abstract}
A patient seller aims to sell a good to an impatient buyer (i.e., one who discounts utility over time). The buyer will remain in the market for a period of time $T$, and her private value is drawn from a publicly known distribution. What is the revenue-optimal pricing-curve (sequence of (price, time) pairs) for the seller? Is randomization of help here? Is the revenue-optimal pricing curve computable in polynomial time? We answer these questions in this paper. We give an efficient algorithm for computing the revenue-optimal pricing curve. We show that pricing curves, that post a price at each point of time and let the buyer pick her utility maximizing time to buy, are revenue-optimal among a much broader class of sequential lottery mechanisms. I.e., mechanisms that allow the seller to post a menu of lotteries at each point of time cannot get any higher revenue than pricing curves. We also show that the even broader class of mechanisms that allow the menu of lotteries to be adaptively set, can earn strictly higher revenue than that of pricing curves, and the revenue gap can be as big as the support size of the buyer's value distribution.
\end{abstract}

\section{Introduction}
The seminal paper of~\citet{Stokey79} introduced the approach of using intertemporal price discrimination as a profitable strategy for the seller, when dealing with buyers who discount future utilities. Time-varying airline-ticket pricing, hotel-room pricing, concert-ticket pricing, ``sales'' in retail pricing are a sampling of the numerous instances in which intertemporal price discrimination is routinely employed. In many of these settings, the seller is more patient than the buyer, i.e., the seller discounts the future utilities less aggressively than the buyers. This aspect of the setting was captured in a followup paper by~\citet{LandsbergerM85}, by allowing for the seller and buyer to discount the future at different rates. The intertemporal price discrimination problem has a rich history in the economics literature (see references of~\citet{LandsbergerM85}). While a lot is known, including the fact that intertemporal pricing strategy is profitable when the buyer's discount rate is higher, three fundamental questions remain open. 1) What is the revenue-optimal pricing curve? (A pricing curve consists of a sequence of prices at a finite number $N_T$ of timestamps $t_1 \leq \cdots \leq t_{N_T}$: $(p(t_1), \ldots, p(t_{N_T}))$. A buyer with value $v$ chooses to buy at his utility-maximizing time stamp, namely $\arg \max_j (v-p(t_j))e^{-t_j}$ (or not buy at all). See \cref{sec:prelim} for a formal definition.) 2) How does the pricing curve's revenue compare with more general mechanisms, including randomized ones like sequential lotteries? 3) Can the revenue-optimal pricing curve be computed efficiently? These questions are interesting both from scientific and commercial points-of-view, and are non-trivial even in the single buyer case. The goal of this paper is to understand this problem in depth.

Concretely, consider the problem of selling an item to a (unit-demand) buyer. The buyer’s private value for the item is drawn from a commonly known distribution. The buyer’s utility decays with time, and is captured by a commonly known discounting factor $\delta(t)$ (the buyer's utility of purchasing at price $p$ at time $t$ is $(v-p)\cdot\delta(t)$), while the seller does not discount future utilities. The buyer remains in the market only for a finite time, from $t = 0$ to $t = T$. The decision problem facing the buyer is whether to spend more to get the item immediately after entering the market, or pay less and get the discounted utility later. Knowing that the buyer faces this tradeoff, what is the seller’s revenue-optimal pricing curve? While the pricing curve, a deterministic object, is the central object of our study in this paper owing to its ubiquitous presence, we also analyze the question of when and whether randomization helps. The goal here is to fully understand the power and limitation of pricing curves.

From a computational point of view as well, as mentioned earlier, the central question remained open: can the revenue-optimal pricing curve be computed in time polynomial in the support-size $|V|$ of the buyer's value distribution? The algorithmic challenge stems from having to jointly compute the optimal timestamps at which to offer the prices, and the optimal prices to offer. The fact that both the timestamps and the prices can be chosen from an uncountable continuum (even though values are drawn from a finite support distribution), calls for making insightful observations to obtain even an exponential-time algorithm.

\paragraph{Our results.} First, we characterize and give a computationally efficient algorithm for the revenue-optimal pricing curve for the seller. Second, we show that this revenue-optimal pricing curve is optimal among the much broader class of randomized mechanisms that let the seller announce at $t = 0$, a menu of lotteries for each time from $t = 0$ to $T$. A lottery menu will consist of a collection of entries, where each entry is a probability of obtaining the item and the price to pay if the item was allocated. This is the most general class of non-adaptive mechanisms possible. Third, we show that the even broader class of mechanisms that let the seller announce adaptive lottery menus, namely, menus designed as a function of which menu option was purchased by the buyer in the past, is strictly more powerful (note that pricing curves are by definition non-adaptive, because, once the buyer purchases at a price, she deterministically receives the item and is out of the market). We show that the gap between adaptive lotteries and pricing curves can nearly be the support size $|V|$ of the buyer's value distribution -- this is as high as the gap can be because, pricing curves can trivially get a $|V|$ approximation to the social welfare.

\paragraph{Challenges and techniques.} Finiteness of total time $T$ is an important source of complication in this problem. Usually in time-discounted settings the total time is taken to be infinite. When the seller does not discount the future, infinite time makes the problem easy because the seller can extract the entire social surplus as revenue.\footnote{For every value in the support of the distribution, the seller can create a (price, time) pair such that only that particular value will buy at this price. This is achieved as follows. At $t=0$ post a price equal to the largest value in the support minus a tiny $\varepsilon$. After a long while, post a price at $t=t'$ equal to the second highest value minus $\varepsilon$. Given that there is a lot to lose via discounting, the highest value in the support will buy at a price of their value minus $\varepsilon$ at $t=0$ instead of waiting till $t'$. Likewise the second highest will buy at $t'$ at a price of second highest value minus $\varepsilon$.} Our results produce the optimal pricing curve for any given value of $T$, and a lot of the technical simplifications afforded by infinite time vanish when the total time is finite.

To compute the optimal pricing curve, we write a mathematical program that captures the expected revenue of the seller in the objective, with IC and IR properties as constraints. The catch is that this program is not an LP or even a convex program. The fact that utility discounting is multiplicative means that, regardless of the exact functional form of the discounting, the program is necessarily non-linear. Instead of solving the program directly, we analyze the program to glean several structural properties of the optimal solution; in particular we obtain properties of the price $p(v)$ at which a buyer with value $v$ in the support of the distribution will buy in the optimal pricing curve. We establish how the prices $p(v_i)$ and $p(v_{i+1})$ of two successive values in the support must be related, and show that they should either be equal, or be related as a function of $v_{i+1}$ and $v_i$. This relation is informative enough to suggest a natural algorithm to compute the pricing curve, albeit running in exponential time: enumerate over all partitions of the support (where a partition consists of a collection of sets of contiguous elements in the support). Given a partition, use the relation established above to obtain the optimal prices for that partition, and return the partition with the optimal revenue. Even in this exponential time algorithm, the process of obtaining the prices given a partition requires more ideas than just the price relation discussed above: for instance, it requires establishing a certain monotonicity that permits us to do a binary search to obtain the prices. From here, we go on to provide a polynomial-time algorithm by developing an efficient method to compute the revenue-optimal partition of values in the support. There are many insights that go into the development, and proof of optimality of the efficient algorithm, including establishing the continuity and monotonicity of the partition functions in time limit $T$, the uniqueness of the optimal pricing curve, etc.

\paragraph{Conceptual contribution and significance.} Apart from the fundamental nature of the optimal pricing curve problem, there is a conceptual contribution in analyzing the power of randomization. Our answer for when and how randomization helps is complete and nuanced. We show that randomization is powerless when the mechanism is forced to be non-adaptive. But coupled with the power of adaptivity, a randomized mechanism can be extremely powerful, earning nearly $|V|$ times higher revenue than any non-adaptive mechanism. This power of adaptive mechanisms lies in the ability to sharply price discriminate in a very short frame of time. Note that even a non-adaptive mechanism like pricing curve can price discriminate (after all, that is what intertemporal price discrimination is all about). In fact, if $T=\infty$, as discussed earlier, a pricing curve can extract the entire social surplus as revenue, and thus there is no revenue gap due to adaptivity given infinite time. \emph{The nuance lies in the fact that non-adaptive mechanisms cannot do very effective price discrimination in short time frames, while adaptive mechanisms can.} And thus when $T$ is finite, the adaptive mechanisms' ability to price discriminate in short time frames is fully exercised leading to dramatic revenue gaps. Our proof that establishes the revenue gap is illuminating: it shows why non-adaptive mechanisms cannot price-discriminate in short time frames.

The fact that deterministic mechanisms are optimal in single-parameter settings is very well known~\citep{Myerson81}. It is also quite well known that randomization affords significantly higher revenue in multi-parameter settings. However, in the mystic twilight of single-parameter settings (value is a single number), with the time element and a discounting factor, the power of randomization was thus far unknown. Our results provide significant understanding of the situation.

\paragraph{Price discrimination as a function of the time horizon $T$.} Our algorithmic result for finding the optimal pricing curve also gives conceptual insights into how the optimal pricing curve depends on the horizon $T$. As we will see in \cref{sec:pricing}, the optimal pricing curve partitions $V$ into groups of adjacent values, and designs a targeted price and time pair for each group. As mentioned earlier, when $T = \infty$, the pricing curve can do perfect price discrimination, and thus the corresponding grouping will have each value in the support in its own group. When $T$ goes down, our algorithmic result shows that the evolution of the grouping has a nice structure: the new grouping is generated by merging adjacent groups in the old grouping. When $T$ goes to 0, the mechanism only sells to one group of values with the same price. This insight is at the core of developing a polynomial-time algorithm for computing the revenue-optimal pricing curve.

%And if this group has size at least 2, the last merging happens exactly at time $0$ (Lemma ??).

\paragraph{Computational benefits over blackbox solvers.} We remark that our approach gives an explicit and simple-to-implement algorithm for computing the optimal pricing curve. As we noted earlier, the straightforward mathematical program for our problem is non-linear. Even after proving some lemmas that simplify the program to get a convex program (namely, (\ref{lp:3})), it is not clear if there is an immediate blackbox poly-time solution to this program. I.e., not all convex programs are polynomial-time solvable (unless P=NP). For example, for feasibility checking using the ellipsoid algorithm it is not clear if a polynomial number of iterations are enough to get the volume of the bounding ellipsoid below the volume of the convex set of feasible solutions. One can indeed get an approximately optimal / approximate feasible solution with the ellipsoid algorithm, but it is unclear beyond that. Further, even if a polynomial-time solution were possible with a convex program solver, the algorithm we develop in Section 4.3 is simpler, and likely to be much faster, than a black-box convex program solver. The structural insights on the optimal solution that lead to this algorithm are interesting on their own, and cannot be obtained by a solver.

\paragraph{Related work comparison.} We begin with the most closely related work to ours, namely, that of \citet{Shneyerov14}. That paper adopts almost the same setting -- the buyers discount the future while the seller does not, and there is a window of $T$ to sell the item. \citet{wang2001optimal} studies a similar setting with an infinite time horizon.

In both of these works, the revenue-maximizing pricing curve is characterized via differential equations. However, there is a crucial difference from our work: We allow arbitrary value distributions (we use finite support in the proof for cleanliness of exposition), while both these works make fairly strong assumptions about them, e.g., having increasing virtual values or the profit function being concave. These assumptions make the monotonicity constraints not binding in our analysis in \cref{sec:pricing}, thus enabling their differential-equation characterizations. On the other hand, for general value distributions, the monotonicity constraints may well be binding (we provide a simple example with this behavior in \cref{sec:pricing}) and it was unclear how to characterize the optimal pricing curve or how to compute it efficiently. To the best of our knowledge, our work is the first to address the computational aspect of the problem by providing a poly-time algorithm for the revenue-optimal pricing curve, for general value distributions. Moreover, our analysis lead to behavioral characterizations for it, where the ``grouping'' behavior in the middle was not seen in these works and poses additional challenge in designing a poly-time algorithm.

Another closely related work is that of \citet{CEP19}, where they study the design of optimal mechanisms in the presence of asymmetric discounting between the seller and the buyers: Both the buyers and seller are discounting the future at different rates, but there is no time limit $T$ to sell the item. As a comparison point, \citet{CEP19} show that the pricing curve they obtain is optimal among the class of universally truthful randomized mechanisms -- this is a restricted class of mechanisms, and for instance, does not include the lottery menu based mechanisms we consider; our proof shows that pricing curves are optimal among the much broader class of all non-adaptive randomized mechanisms, including those that are truthful only in expectation. Also, as discussed earlier we precisely pinpoint the class of mechanisms whose revenue can be achieved by pricing curves, and a meaningful class that obtains higher revenue.

The work of~\citet{Briceno-AriasCP17} also studies the problem of computing the optimal pricing curves with a random number of buyers arriving over time. There are two differences from our work. First, their paper primarily addresses the case where there are two buyer types, and the extension to arbitrary number of types requires an assumption. Second, their Poisson model for random arrival of buyers over time is different from our single buyer model. In particular, there is competition between buyers in their model, all competing for a single item. Our model can either be thought of as a single buyer, or as an infinite supply (digital goods) model, where there is no competition between buyers. This makes our mathematical problems quite different.

\paragraph{Additional related work.}
The problem of selling optimally over a period of time to strategic buyers is a classic problem in economics and operations research. It has been studied in various settings  \citep{Stokey79,Stokey81, Conlisk84,LandsbergerM85,BesankoW90,borgs2014optimal,Shneyerov14,BesbesL15,Briceno-AriasCP17,CEP19,CorreaPV20}. For example, \citet{Stokey79} shows that when the buyer and the seller have the same discounting factor, the optimal pricing curve is to sell at time $0$. \citet{LandsbergerM85} study when the optimal pricing curve is posting a single price, and when to price discriminate. \citet{BesbesL15} and \citet{borgs2014optimal} consider the problem with discrete time slots and non-discounting buyers purchasing the lowest price of different time windows. \citet{Briceno-AriasCP17} work on the characterization of the optimal pricing curves when the impatient buyers arrive randomly over time. 

Another related problem is the FedEx problem studied by~\cite{FGKK16}. There the buyer has a private value and a deadline (there are $m$ possible deadlines) jointly drawn from a distribution, and receives a value only when the allocation happens before the deadline. The common theme with our paper is that the buyer is time-sensitive, i.e., has $0$ value for receiving an allocation beyond a certain time. But there are many differences, making the problem incomparable with ours. First, in our model, the value continuously decays with time, whereas in the FedEx problem the value remains the same until the deadline and goes to zero afterwards. Second, in the FedEx problem, only the prices are chosen, as the allocation can happen only in one of finitely many timestamps. In our problem, the prices and the timestamps have to be jointly optimized over, and chosen from an uncountable continuum. Third, in the FedEx problem the value and deadline are both private parameters, while in our setting buyer's value is the only private parameter.

The problem of selling optimally over time has also been shown empirically relevant, in the contexts of video games sales \citep{Nair07}, retail sales in supermarkets \citep{Pesendorfer02}, and airline tickets sales \citep{LiGN14}.

Our efficient algorithm is partly inspired by the idea of optimization by the continuation method (a.k.a. homotopy method), which is a general and problem independent technique for tackling nonconvex
problems. Intuitively, it starts with an objective function that is easy to solve (e.g. convex function), and progressively transforms it to the required objective~\citep{mobahi2015link}. Throughout
this progression, the solution of each intermediate objective is used as a starting point to search for the solution of the next one. It can be used to compute fixed points~\citep{eaves1972homotopies}, and has been playing an important role in economic research~\citep{eaves1999general,herings2010homotopy}. Recently, the homotopy method has been applied in multiple learning tasks, such as tensor PCA~\citep{anandkumar2017homotopy} and various vision applications~\citep{nikolova2010fast,mobahi2012seeing}.

\section{Preliminaries}
\label{sec:prelim}
\subsection{Bayesian Setting for an Impatient Buyer}
We consider the problem of selling an item to an impatient buyer who discounts the future over a time horizon $[0, T]$. We assume the buyer's value is drawn from a probability mass function $f$ over a discrete set $V = \{v_1, v_2, \ldots, v_n\}$ with $0 \leq v_1 < v_2 < \cdots < v_n$. This is a standard assumption when computation is involved~\citep{cai2016duality,chawla2007algorithmic,chawla2010multi,chawla2015power,hart2017approximate,li2013revenue,babaioff2020simple}; and we show in \cref{sec:continuous} that all of our results can be made arbitrarily close to the optimum for continuous distributions via discretization. 

We assume the buyer discounts the future with rate $e^{-1}$: if she has value $v$ and buys at time $t$ for a price of $p$, her utility will be $(v - p) \cdot e^{-t}$. We note that any continuously decreasing discounting function $\delta' : [0, T'] \rightarrow (0, 1]$ is equivalent to the exponential discounting $\delta(t) = e^{-t}$. In fact, any mechanism providing a menu at time $t'$ for a buyer with discounting function $\delta'$ is equivalent to providing the same menu at time $t = \delta^{-1}(\delta'(t')) = -\ln \delta'(t')$ for a buyer with discounting function $\delta$, if the time limit $T = \delta^{-1}(\delta'(T')) = -\ln \delta'(T')$. 
We further assume the buyer is risk-neutral and she cares about her expected utility when randomness is present. On the other hand, the seller is perfectly patient, and wishes to maximize his revenue. The transaction can only be done in time $[0, T]$ and the buyer will lose the interest of purchasing the item after the time limit $T$. Henceforth, an instance of our problem is specified by a tuple $\langle T, V, f \rangle$.

\subsection{Non-adaptive Sequential Lottery Mechanisms and Pricing Curves}
Pricing curves are the central object of study in this paper, and they are a special case of non-adaptive \emph{sequential lottery mechanisms} (we drop non-adaptive sequential when clear from context). In this mechanism, the seller picks a finite number $\nt$ of timestamps $t_1 \leq t_2 \leq \cdots \leq t_{\nt}$, and posts a menu $M_i$ of lotteries at each timestamp $t_i$. Each menu $M_i$ consists of several options / lotteries of the form ``allocate with probability $x_{ij}$; if allocated, the payment,\footnote{For risk-neutral buyers, lotteries charging only upon allocation are completely equivalent to lotteries that charge always. An always-charging-lottery, priced at $p$, offering a value-$v$ item with probability $q$ yields the buyer a utility of $vq-p$. Equivalently, a lottery that charges $p/q$ only upon success will yield the buyer a utility of $q(v-p/q)=vq-p$.} is $p_{ij}$''. The buyer is asked to pick exactly one option from each menu, where every menu includes a null option with $0$ allocation probability and $0$ price. The risk-neutral buyer sees the entire sequence of menus at $t=0$, and computes her optimal solution, which would be of the form ``At every timestamp $t_i$, do the following: if no item has been allocated from any of the previously picked options / lotteries at earlier timestamps, pick option \emph{optimal-option} at $t_i$''. The mechanism proceeds in the following straightforward manner: the buyer inspects each timestamp $t_i$ in sequence, and picks her precomputed optimal option from the menu at $t_i$. After picking the lottery, the buyer observes the realized outcome of the lottery immediately. If she gets the item, the mechanism ends. If not, the buyer inspects the next timestamp $t_{i+1}$.

Formally: the buyer learns her value $v$ and the entire sequence of menus at $t=0$. At each timestamp $t_i$:
\begin{itemize}
   \item If the buyer picks the $(x_{ij}, p_{ij})$ option, the buyer is allocated with probability $x_{ij}$.
        \item If allocated, the mechanism ends. The buyer is charged $p_{ij}$, accruing utility of $e^{-t_i} \cdot (v - p_{ij})$.
        \item If not allocated, the buyer proceeds to inspect timestamp $t_{i+1}$.
\end{itemize}
The buyer's objective is to optimize her expected utility. Let $u_i(v)$ be the buyer's continuation utility starting from timestamp $t_i$ when the buyer's value is $v$:
\[
    u_i(v) = \max_j x_{ij} \cdot e^{-t_i} \cdot (v - p_{ij}) + (1 - x_{ij}) \cdot u_{i+1}(v)
\]
For the last timestamp $t_{N_T}$, we have $u_{\nt}(v) = \max_j e^{-t_{\nt}} \cdot x_{\nt j} \cdot (v - p_{\nt j})$.

Note that a sequential lottery mechanism is the most general non-adaptive mechanism possible for a single buyer.

\paragraph{Pricing curve.} A pricing curve or a \emph{sequential pricing mechanism} is a deterministic lottery mechanism, which posts a single $p_i$ at any timestamp $t_i$, for a deterministic allocation (i.e., $x_i = 1$). Given a pricing curve, the buyer can choose to buy at a utility maximizing time $t_i$, or not to buy at all.

\paragraph{Why allow only finite number of timestamps?} With a distribution of finite support, pricing curves need at most $|V|$ timestamps, and thus finiteness of timestamps is without loss of generality. We show in \cref{sec:lottery} that for any finite number of timestamps, the revenue from sequential lottery mechanisms is no higher than that of pricing curves. Thus, the limit of the optimal revenue from sequential lottery mechanisms, as the number of timestamps goes to infinity, is the optimal revenue from pricing curves.

\paragraph{Exposing the whole pricing curve.} The buyer in our model knows the entire pricing curve, or the entire sequence of lottery menus, at $t=0$. This captures the fact that retail sales are often announced weeks in advance; airline pricing websites like Kayak typically give guidance on how prices are expected to move.

\subsection{Adaptive Lotteries}
An adaptive sequential lottery mechanism is just like a non-adaptive sequential lottery mechanism, except that the set of options available to purchase at each timestamp $t_i$ can be a function of the history of past purchases. The null option should be accessible to the buyer at each menu regardless of history --- otherwise, the seller can simulate loans by allocating at $t=0$ and charging $ve^T$ at $t=T$ (note that the seller does not time discount), which yields unrealistically high revenue (larger than even social welfare). %With the null option, we show that the optimal adaptive lottery's revenue is upper bounded by social welfare (\cref{thm:revenuewelfarebound}).

\section{Pricing Mechanisms Are Optimal Non-adaptive Mechanisms}
\label{sec:lottery}
In this section, we show that pricing curves are as powerful as non-adaptive sequential lotteries. Non-adaptive sequential lotteries are without loss of generality the most general non-adaptive mechanisms, since by taxation principle, the mechanism at each timestamp $t_i$ is equivalent to a menu of $n$ options $(x_{ij}, p_{ij})$ corresponding to the allocation probability and expected price when allocated for $v_j$ in that single-timestamp mechanism. In \cref{sec:adaptive} we show how adaptive randomized mechanisms can get strictly higher revenue by cleverly exploiting price discrimination in relatively short time frames.
%Note that a sequential lottery mechanism is the most general non-adaptive mechanism possible for a single buyer.
\begin{theorem}
\label{thm:pricing_optimal}
For any instance $\langle T, V, f \rangle$ of the problem, the revenue obtained by the optimal non-adaptive lottery mechanism can also be achieved by a pricing mechanism.
\end{theorem}
For this purpose, we introduce (sequential) {\em single-lottery mechanisms}, in which there is exactly one option (in addition to the null option) in the menu at each timestamp. Multiple single option menus are allowed at the same timestamp. We show that single-lottery mechanisms are as powerful as lottery mechanisms; and moreover, pricing mechanisms are as powerful as single-lottery mechanisms. %At the end of this section, we demonstrate with an example that, however, adaptive lottery mechanisms are more powerful than pricing mechanisms.

\begin{lemma}
\label{lem:single_lottery}
Any non-adaptive lottery mechanism can be simulated by a single-lottery mechanism.
\end{lemma}
\begin{proof}
Suppose menu $M$ at timestamp $t$ has $k$ options with $x_1 > \cdots > x_k > x_{k+1} = 0$. We assume that $p_1 > \cdots > p_k > p_{k+1} = 0$. This is without loss of generality, since if $x_i \geq x_j$ but $p_i \leq p_j$, option $j$ will never be chosen.
Moreover, it is without loss of generality to assume that $x_i p_i$ as a function of $x_i$ is convex. If not, for the purpose of contradiction, assume that there exist $i < \ell < j$ with $x_{\ell} = \lambda x_i + (1 - \lambda) x_j$ but $x_{\ell} p_{\ell} > \lambda x_i p_i + (1 - \lambda) x_j p_j$. However, option $\ell$ will never be chosen, as choosing one of $(x_i, p_i)$ and $(x_j, p_j)$ would be better. Precisely, let $u_C$ be the buyer's expected utility starting from the next timestamp and we have:
\begin{align*}
&~x_{\ell} (v - p_{\ell}) \cdot e^{-t} + (1 - x_{\ell}) u_C\\
< &~\lambda \cdot \big(x_i (v - p_i) \cdot e^{-t} + (1 - x_i) u_C\big) + (1 - \lambda) \cdot \big(x_j (v - p_j) \cdot e^{-t} + (1 - x_j) u_C\big)\\
\leq &~ \max\big\{x_i (v - p_i) \cdot e^{-t} + (1 - x_i) u_C, \ x_j (v - p_j) \cdot e^{-t} + (1 - x_j) u_C\big\}.
\end{align*}

For menu $M$ at timestamp $t$, we create $k$ timestamps in the order of $t_{1}, t_{2}, \ldots, t_{k}$ at time $t$ for the single-lottery mechanism. At timestamp $t_{i}$, let $x'_i$ be the allocation probability and $p'_i$ be the unit price of the lottery, and we set
\[
    x'_i = 1 - \frac{1 - x_i}{1 - x_{i + 1}} \quad \mbox{and} \quad p'_i = \frac{x_i p_i  - (1 - x'_i)x_{i + 1} p_{i + 1}}{x'_i}.
\]
Conceptually, these are marginal allocation probabilities and prices. Naturally, we have $p'_i \geq p'_{i+1}$ for all i:
\begin{align*}
p'_{i} - p'_{i + 1}
= \frac{x_i p_i (x_{i + 1} - x_{i + 2}) + x_{i + 2} p_{i + 2} (x_i - x_{i + 1}) - x_{i + 1}p_{i + 1}(x_i - x_{i + 2})}{(x_i - x_{i + 1})(x_{i + 1} - x_{i + 2})/(1 - x_{i + 1})} \geq 0,
\end{align*}
where the inequality comes from the convexity of $x_i p_i$ with respect to $x_i$.

Since the prices are decreasing, it turns out that if a buyer chooses $(x'_z, p'_z)$, she will always choose the lotteries after it, i.e., $\{(x'_z, p'_z), (x'_{z + 1}, p'_{z + 1}), \ldots, (x'_k, p'_k)\}$. Suppose the buyer chooses $(x'_z, p'_z)$ but not $(x'_{z+1}, p'_{z+1})$. Observe that the buyer with value $v$ chooses an option $(x'_z, p'_z)$ but not $(x'_{z+1}, p'_{z+1})$ if and only if
\[
    e^{-t} \cdot x'_z \cdot (v - p'_z) + (1 - x'_z) \cdot u_C \geq u_C,
\]
which is equivalent to
$e^{-t} \cdot (v - p'_z) \geq u_C$, where $u_C$ is the continuation utility starting from the timestamp $t_{z+2}$. However, the buyer should also choose the next option $(x'_{z+1}, p'_{z+1})$ since
\[
    e^{-t} \cdot x'_{z+1} \cdot (v - p'_{z+1}) + (1 - x'_{z+1}) \cdot u_C \geq e^{-t} \cdot x'_{z+1} \cdot (v - p'_{z}) + (1 - x'_{z+1}) \cdot u_C \geq u_C.
\]
Therefore, once the buyer chooses $(x'_z, p'_z)$, she will choose the lotteries after it.

We conclude the proof by showing that choosing the option $(x_z, p_z)$ in the original lottery mechanism is equivalent to choosing a collection of options $\{(x'_z, p'_z), (x'_{z + 1}, p'_{z + 1}), \ldots, (x'_k, p'_k)\}$ in the single-lottery mechanism.

Note that the probability of getting the item by choosing $\{(x'_z, p'_z), (x'_{z + 1}, p'_{z + 1}), \ldots, (x'_k, p'_k)\}$ is
\[
1 - \prod_{i = z}^k (1 - x'_i) = 1 - \left(\prod_{i = z}^{k - 1} \frac{1 - x_i}{1 - x_{i + 1}}\right) \cdot (1 - x_k) = x_z,
\]
and its expected payment is given by:
\begin{align*}
\sum_{i = z}^k x'_i \cdot p'_i \cdot \prod_{j = z}^{i - 1} (1 - x'_j) &= \sum_{i = z}^k x'_i \cdot p'_i \cdot \frac{1 - x_z}{1 - x_i}\\
&= (1 - x_z) \cdot \sum_{i = z}^k \left(x_i p_i  - (1 - x'_i)x_{i + 1} p_{i + 1}\right) \cdot \frac{1}{1 - x_i}\\
&= (1 - x_z) \cdot \sum_{i = z}^k \left(\frac{x_i p_i}{1 - x_i}  - \frac{x_{i + 1} p_{i + 1}}{1 - x_{i + 1}}\right)\\
&= x_z \cdot p_z.
\end{align*}
We can then apply the transformation for all timestamps to obtain a single-lottery mechanism.
\end{proof}

By \cref{lem:single_lottery}, we can now focus on single-lottery mechanisms to finish the reduction for proving \cref{thm:pricing_optimal}. Given any single-lottery mechanism, we will perform a procedure to \emph{derandomize} it into a distribution over pricing mechanisms. 

We denote by $(x_i, p_i)$ the unique option at timestamp $t_i$ in a single-lottery mechanism. For a timestamp in which the unique option is never chosen, we can simply remove this timestamp. Assume there are totally $k$ timestamps in the given single-lottery mechanism. Let $\ell_i$ be the minimum value with which a buyer will choose the unique option at timestamp $t_i$.
\begin{lemma}
\label{lem:maximum-value-to-purchase}
    At timestamp $t_i$, we have $(\ell_i - p_i) \cdot e^{-t_i} = u_{i + 1}(\ell_i)$ and moreover,
    \[
    \begin{cases}
        (v - p_i) \cdot e^{-t_i} \geq u_{i + 1}(v) & \mbox{if}~ v \geq \ell_i \\
        (v - p_i) \cdot e^{-t_i} < u_{i + 1}(v) & \mbox{if}~ v < \ell_i
    \end{cases}
    \]
    In other words, any value $v \geq \ell_i$ will purchase the lottery at timestamp $t_i$ while any value $v < \ell_i$ will not purchase the lottery at timestamp $t_i$.
\end{lemma}
\begin{proof}
    Notice that $\ell_i$ is indifferent between choosing the unique option at timestamp $t_i$ and skipping:
    \[
        e^{-t_i} \cdot x_i \cdot (\ell_i - p_i) + (1 - x_i) \cdot u_{i + 1}(\ell_i) = u_{i + 1}(\ell_i) 
    \]
    which is $(\ell_i - p_i) \cdot e^{-t_i} = u_{i + 1}(\ell_i)$. As for $v > \ell_i$, observe that we have
    \[
        u_{i + 1}(\ell_i) \geq u_{i + 1}(v) - (v - \ell_i) \cdot e^{-t_{i}},
    \]
    since a buyer with value $\ell_i$ can take the options as if she had value $v$ and the allocation probability is at most $1$. Therefore, we have
    \[
        (v - p_i) \cdot e^{-t_i} = (v - \ell_i) \cdot e^{-t_i} + (\ell_i - p_i) \cdot e^{-t_i} = (v - \ell_i) \cdot e^{-t_i} + u_{i + 1}(\ell_i) \geq u_{i + 1}(v).
    \]
    where the second equality follows $(\ell_i - p_i) \cdot e^{-t_i} = u_{i + 1}(\ell_i)$. 
    Finally, by the definition of $\ell_i$, for any $v < \ell_i$, the buyer prefers to skip the option, and therefore,
    \[
        e^{-t_i} \cdot x_i \cdot (v - p_i) + (1 - x_i) \cdot u_{i + 1}(v) < u_{i + 1}(v) 
    \]
    which implies $(v - p_i) \cdot e^{-t_i} < u_{i + 1}(v)$.
\end{proof}

Let $R_i$ be a Bernoulli random variable such that $R_i = 1$ with probability $x_i$. We denote by $r_i \in \{0, 1\}$ the realization of $R_i$. For each $(r_1, \ldots, r_k) \in \{0, 1\}^k$, with probability $\prod_{i = 1}^k x_i^{r_i} (1 - x_i)^{1 - r_i}$, we create a pricing mechanism with pricing function $p'$. The price $p'_i$ at timestamp $t_i$ maps $(r_i, \ldots, r_k) \in \{0, 1\}^{k-i+1}$ to a price. Moreover, let $u'_i(v; r_i, \ldots, r_k)$ be the buyer's continuation utility starting at timestamp $t_i$ given $r_i, \ldots, r_k$ when her value is $v$. $u'_i$ and $p'_i$ are jointly defined in a recursive manner such that
\[
    u'_i(v; r_i, \ldots, r_k) = \max\left\{e^{-t_i} \cdot \big(v - p'_i(r_i, \ldots, r_k)\big), \ u'_{i + 1}(v; r_{i+1}, \ldots, r_k)\right\}.
\]
with $u'_{k+1}(v) = 0$; and $p'_i$ is defined as
\[
    p'_i(r_i, \ldots, r_k) = 
    \begin{cases}
        \infty & \mbox{if}~r_i = 0 \\
        \ell_i - e^{t_i} \cdot u'_{i+1}(\ell_i; r_{i+1}, \ldots, r_k) & \mbox{if}~r_i = 1
    \end{cases}
\]
Intuitively, the item is not sold at timestamp $t_i$ by setting the price to $\infty$ if $r_i = 0$; and if $r_i = 1$, then according to \cref{lem:maximum-value-to-purchase}, $p'_i(r_i, \ldots, r_k)$ is set to be the maximum price so that a buyer with value $\ell_i$ will purchase at timestamp $t_i$.

% , when the future price 
% which depends on the realizations of $R_{i + 1}, R_{i + 2}, \ldots, R_k$. Let $u'_i(v; R_i, R_{i + 1}, \ldots, R_k)$ be the expected utility of a buyer of value $v$ in the pricing mechanism, if she starts at timestamp $t_i$. Then $p'_i(r_{i + 1}, r_{i + 2}, \ldots, r_k)$ is given by:
% \[
% (\ell_i - p'_i(r_{i + 1}, r_{i + 2}, \ldots, r_k)) \cdot e^{-t_i} = u'_{i + 1}(\ell_i; r_{i + 1}, r_{i + 2}, \ldots, r_k).
% \]
% This procedure is illustrated in \cref{alg:derandomization}.

% \begin{algorithm}[H]
% 	\DontPrintSemicolon
	
% 	\ForEach{timestamp $t_i$} {
% 	    Calculate $\ell_i$ using $(\ell_i - p_i) \cdot e^{-t_i} = u_{i + 1}(\ell_i)$\;
% 	}
% 	Randomly realize $R_1 = r_1, R_2 = r_2, \ldots, R_k = r_k$, where $R_i = 1$ with probability $x_i$\;
	
% 	\ForEach{timestamp $t_i$ with $r_i = 1$ from $i = k$ back to $i = 1$} {
% 	    Calculate $p'_i(r_{i + 1}, r_{i + 2}, \ldots, r_k)$ using $(\ell_i - p'_i(r_{i + 1}, r_{i + 2}, \ldots, r_k)) \cdot e^{-t_i} = u'_{i + 1}(\ell_i; r_{i + 1}, r_{i + 2}, \ldots, r_k)$\;
% 	    Create a deterministic option at $t_i$ with price $p'_i(r_{i + 1}, r_{i + 2}, \ldots, r_k)$\;
% 	}
% 	\caption{Derandomization}\label{alg:derandomization}
% \end{algorithm}

\definecolor{inv}{rgb}{0.70, 0.70, 0.70}
\definecolor{win}{rgb}{0.5, 0.95, 0.5}
\begin{table}
\centering
 \begin{tabular}{|| c | c | c | c ||} 
 \hline
 $(r_1, r_2, r_3)$ & $p'_1$ & $p'_2$ & $p'_3$ \\ [0.5ex] 
 \hline\hline
 $(0, 0, 0)$ & \cellcolor{inv}$\infty$ & \cellcolor{inv}$\infty$ & \cellcolor{inv}$\infty$ \\ 
 \hline
 $(0, 0, 1)$ & \cellcolor{inv}$\infty$ & \cellcolor{inv}$\infty$ & $4$ \\ 
 \hline
 $(0, 1, 0)$ & \cellcolor{inv}$\infty$ & $8$ & \cellcolor{inv}$\infty$ \\ 
 \hline
 $(0, 1, 1)$ & \cellcolor{inv}$\infty$ & $6$ & $4$ \\ 
 \hline
 $(1, 0, 0)$ & $16$ & \cellcolor{inv}$\infty$ & \cellcolor{inv}$\infty$ \\ 
 \hline
 $(1, 0, 1)$ & $13$ & \cellcolor{inv}$\infty$ & $4$ \\ 
 \hline
 $(1, 1, 0)$ & $12$ & $8$ & \cellcolor{inv}$\infty$ \\ 
 \hline
 $(1, 1, 1)$ & $11$ & $6$ & $4$ \\ 
 \hline
\end{tabular}
\caption{All possible derandomized mechanisms in \cref{ex:derandomization}}
\label{tab:derandomization}
\end{table}

\begin{example}
\label{ex:derandomization}
Let $k = 3$ and assume we have a single-lottery mechanism with $(x_1, p_1, t_1) = (0.5, 13, 0)$, $(x_2, p_2, t_2) = (0.5, 7, \ln 2)$, $(x_3, p_3, t_3) = (0.5, 4, 2 \ln 2)$. Using \cref{lem:maximum-value-to-purchase}, we can compute that $\ell_3 = 4$, $\ell_2 = 8$, and $\ell_1 = 16$. %For example, $\ell_2$ is calculated using $(\ell_2 - p_2) \cdot e^{-t_2} = u_3(\ell_2) = \max(0, \ x_3(\ell_2 - p_3) \cdot e^{-t_3})$. 
For this single-lottery mechanism, each $(r_1, r_2, r_3) \in \{0, 1\}^3$ is sampled with probability $1/8$; and \cref{tab:derandomization} shows the pricing mechanisms for all possible combinations of $(r_1, r_2, r_3)$.
\end{example}
%  For example, when $r_1 = 0, r_2 = r_3 = 1$, $p'_2(r_3)$ is calculated so that a buyer with value $\ell_2$ will narrowly buy this option: $(\ell_2 - p'_2(r_3)) \cdot e^{-t_2} = u'_3(\ell_2; r_3) = \max(0, (\ell_2 - p_3') \cdot e^{-t_3})$, giving $\frac{1}{2}(8 - p'_2(r_3)) = 1$ and thus $p'_2(r_3) = 6$. \cref{tab:derandomization} shows all possible results of \cref{alg:derandomization}.

%The result of \cref{alg:derandomization} is a distribution of deterministic mechanisms. 
We claim the expected revenue over all pricing mechanisms that are created is equal to the revenue of the given single-lottery mechanism. In fact, for any value $v$ and any timestamp $t_i$, the expected payment at $t_i$ of a buyer with value $v$ over all pricing mechanisms is equal to that of the single-lottery mechanism. Let $w'_i(v) = \E[u'_i(v; R_i, \ldots, R_k)]$ be the buyer's expected utility, where the expectation is taken over $(R_i, \ldots, R_k)$.
\begin{lemma}
For all $i$, $p_i = \E[p'_i(1, R_{i+1}, \ldots, R_k)]$ and $u_i(v) = w_i'(v) = \E[u'_i(v; R_i, \ldots, R_k)]$.
\label{lem:derand_same}
\end{lemma}
\begin{proof}
We prove by a backward induction from $i = k$ back to $i = 1$. As the base case where $i = k$, we simply have $p'_k(1) = p_k(1) = \ell_k$. As for the utilities, if $v < \ell_k$, then $u_k(v) = w_k'(v) = 0$; and if $v > \ell_k$, we have
\[
w_k'(v) = \E[u'_k(v; R_k)] = x_k \E[u'_k(v; 1)] + (1 - x_k) \E[u'_k(v; 0)] = e^{-t_k} \cdot x_k \cdot (v - p'_k)  + 0 = u_k(v).
\]

For the inductive step, assume the lemma statement holds for $i + 1$. Recall that by \cref{lem:maximum-value-to-purchase}, $(\ell_i - p_i) \cdot e^{-t_i} = u_{i + 1}(\ell_i)$, and moreover, by the construction of $p'_i(1, r_{i + 1}, \ldots, r_k)$, we have
\[
p'_i(1, r_{i + 1}, \ldots, r_k) = \ell_i - u'_{i + 1}(\ell_i; r_{i + 1}, r_{i + 2}, \ldots, r_k) \cdot e^{t_i}
\]
for any $r_{i + 1}, r_{i + 2}, \ldots, r_k$. Therefore,
\begin{align*}
\E[p'_i(1, R_{i + 1}, \ldots, R_k)] &= \ell_i - \E[u'_{i + 1}(\ell_i; R_{i + 1}, R_{i + 2}, \ldots, R_k)] \cdot e^{t_i}\\
&= \ell_i - w'_{i + 1}(\ell_i) \cdot e^{t_i} = \ell_i - u_{i + 1}(\ell_i) \cdot e^{t_i} = p_i,
\end{align*}
where the third equality applies the induction hypothesis. As for the utilities, if $v < \ell_i$, then we simply have $u_i(v) = u_{i + 1}(v)$ and $w'_i(v) = w'_{i + 1}(v)$, leading to $u_i(v) = w'_i(v)$. Otherwise,
\begin{align*}
w'_i(v) &= \E[u'_i(v; R_i, R_{i + 1}, \ldots, R_k)]\\
&= x_i\E[u'_i(v; 1, R_{i + 1}, \ldots, R_k)] + (1 - x_i)\E[u'_i(v; 0, R_{i + 1}, \ldots, R_k)]\\
&= x_i (v - \E[p'_i(R_{i + 1}, \ldots, R_k)]) \cdot e^{-t_i} + (1 - x_i) w'_{i + 1}(v)\\
&= x_i (v - p_i) \cdot e^{-t_i} + (1 - x_i) u_{i + 1}(v)\\
&= u_i(v).
\end{align*}
This finishes the induction and shows the expected prices and utilities in the distribution of pricing mechanisms are the same as those in the sequential single-lottery mechanism.
\end{proof}
We are now ready to combine \cref{lem:single_lottery}, \cref{lem:maximum-value-to-purchase}, \cref{lem:derand_same} to prove \cref{thm:pricing_optimal}.
\begin{proof}[Proof of \cref{thm:pricing_optimal}]
Given any single-lottery mechanism, we can apply the derandomization process to obtain a distribution of pricing mechanisms. In each pricing mechanism, since $\ell_i$, the minimum value to purchase the option at timestamp $t_i$, remains the same by the construction, the decision of whether to choose the option or not at timestamp $t_i$ also remains the same for other values by \cref{lem:maximum-value-to-purchase}. As a result, for any value $v$, the probability of the buyer with value $v$ reaching timestamp $t_i$ over all pricing mechanisms is the same as that of the single-lottery mechanism. Since the expected prices at each timestamp are also the same by \cref{lem:derand_same}, the expected revenue over the pricing mechanisms is the same as the single-lottery mechanism. Therefore, there exists a pricing mechanism achieving at least the same revenue.
\end{proof}
%\kn{Does this hold for multiple discount factor?}

\section{Characterizations of Optimal Pricing Mechanisms}
\label{sec:pricing}

In this section, we provide characterizations for the structure of the optimal pricing mechanism. Moreover, we develop our algorithm to compute it that runs in time polynomial in $|V|$.

\subsection{Formulation as a Mathematical Program}
We begin with formulating the computation of the optimal pricing mechanism as a mathematical program. Let $p(v_i)$ and $t(v_i)$ be the price and time for a buyer with value $v_i$ to purchase the item. If a buyer with value $v_i$ buys the item at time $t(v_i)$, then we have $\big(v_i - p(v_i)\big) \cdot e^{-t(v_i)} \geq 0$. As a result, any buyer with value $v_j > v_i$ can buy the item at time $t(v_i)$ to achieve non-negative utility since 
\[
    \big(v_j - p(v_i)\big) \cdot e^{-t(v_i)} > \big(v_i - p(v_i)\big) \cdot e^{-t(v_i)} \geq 0.
\]
This fact allows us to enumerate the minimum valuation $v^{\mathrm{min}}$ that participates in the auction and compute the optimal pricing mechanism conditioned on each possible minimum valuation $v^{\mathrm{min}}$. Without loss of generality, we rename $v_1 < v_2 < \cdots < v_n$ to be the valuations participating in the auction. We assume $n \geq 2$ since when $n = 1$, one can simply charge $v_1$ at time $0$ to achieve optimality. To maximize revenue, the seller solves the following mathematical program:
\begin{equation}
\tag{A}
\begin{array}{ll@{}ll}
\text{maximize}  & \displaystyle\sum\limits_{i = 1}^n p(v_i) f(v_i) &\\
\text{subject to}& \displaystyle\big(v_i - p(v_i)\big) \cdot e^{-t(v_i)} \geq \big(v_i - p(v_{j})\big) \cdot e^{-t(v_{j})}, \quad \quad \forall i \neq j \in [n], \\
&v_i - p(v_i) \geq 0, \quad \quad \forall i \in [n], \\
&t(v_i) \in [0,T], \quad \quad \forall i \in [n].
\end{array}
\label{lp:1}
\end{equation}
Here $\big(v_i - p(v_i)\big) \cdot e^{-t(v_i)} \geq \big(v_i - p(v_{j})\big) \cdot e^{-t(v_{j})}$ is the incentive compatibility (IC) constraint ensuring that the buyer with value $v_i$ does not switch to other options; and $v_i - p(v_i) \geq 0$ corresponds to the individual rationality (IR) constraint ensuring that the buyer does not incur negative utility by choosing the designated option. It turns out that many of the IC constraints are redundant and we can simplify the program as:
\begin{equation}
\tag{B}
\begin{array}{ll@{}ll}
\text{maximize}  & \displaystyle\sum\limits_{i = 1}^n p(v_i) f(v_i) &\\
\text{subject to}& \displaystyle\big(v_i - p(v_i)\big) \cdot e^{-t(v_i)} \geq \big(v_i - p(v_{i - 1})\big) \cdot e^{-t(v_{i - 1})}, \quad \quad \forall i \in [n] \setminus \{1\}, \\
&v_i - p(v_i) \geq 0, \quad \quad \forall i \in [n], \\
&p(v_i) - p(v_{i - 1}) \geq 0, \quad \quad \forall i \in [n] \setminus \{1\}, \\
&t(v_i) \in [0,T], \quad \quad \forall i \in [n].
\end{array}
\label{lp:2}
\end{equation}
Intuitively, by adding the monotonicity constraint on the price, it suffices to check whether a buyer with value $v_i$ has incentive to switch to the option designated to value $v_{i-1}$ only.

\begin{proposition}
\label{prop:program_equiv}
Any optimal solution of Program~\eqref{lp:2} is an optimal solution of Program~\eqref{lp:1}.
\end{proposition}
\begin{proof}
Note that any IC constraint $\big(v_i - p(v_i)\big) \cdot e^{-t(v_i)} \geq \big(v_i - p(v_{j})\big) \cdot e^{-t(v_{j})}$ with $i < j$ is redundant. This is because if there exists a value $v_i$ that prefers to switch to the option designated to $v_j > v_i$, then we can change the option in the solution for $v_i$ to $\big(p(v_j), t(v_j)\big)$ to increase the objective without violating any constraint.

As we introduce the constraints of price monotonicity $p(v_i) - p(v_{i - 1}) \geq 0$ in Program~\eqref{lp:2}, we next show that any optimal solution $(p^A, t^A)$ of Program~\eqref{lp:1} satisfies $p^A(v_i) - p^A(v_{i - 1}) \geq 0$ for all $i$. For any $i < j$, $(p^A, t^A)$ must satisfy the IC constraints
$\big(v_j - p^A(v_j)\big) \cdot e^{-t^A(v_j)} \geq \big(v_j - p^A(v_{i})\big) \cdot e^{-t^A(v_{i})}$ and $\big(v_i - p^A(v_i)\big) \cdot e^{-t^A(v_i)} \geq \big(v_i - p^A(v_{j})\big) \cdot e^{-t^A(v_{j})}$.
%\footnote{The upward IC constraint $\big(v_i - p^A(v_i)\big) \cdot e^{-t^A(v_i)} \geq \big(v_i - p^A(v_{j})\big) \cdot e^{-t^A(v_{j})}$ is still automatically satisfied after we removed them.}
Taking their product,
\[
\big(v_i - p^A(v_i)\big) \cdot \big(v_j - p^A(v_j)\big) \geq \big(v_i - p^A(v_j)\big) \cdot \big(v_j - p^A(v_i)\big),
\]
which implies $p^A(v_i) \leq p^A(v_j)$ for $i < j$.

Finally, we show that only IC constraints for adjacent values are needed. Consider an optimal solution $(p^B, t^B)$ of Program~\eqref{lp:2}. For any $j = i+1$ and $k=i+2$, $(p^B, t^B)$ must satisfy:
\[
\big(v_k - p^B(v_k)\big) \cdot e^{-t^B(v_k)} \geq \big(v_k - p^B(v_j)\big) \cdot e^{-t^B(v_j)}
\]
and
\[
\big(v_j - p^B(v_j)\big) \cdot e^{-t^B(v_j)} \geq \big(v_j - p^B(v_i)\big) \cdot e^{-t^B(v_i)}.
\]
Moreover, since $p^B(v_j) \geq p^B(v_i)$, we have $t^B(v_j) \leq t^B(v_i)$. Therefore, we have
\begin{align*}
\big(v_k - p^B(v_k)\big) \cdot e^{-t^B(v_k)} &\geq \big(v_k - p^B(v_j)\big) \cdot e^{-t^B(v_j)}\\
&= \big(v_j - p^B(v_j)\big) \cdot e^{-t^B(v_j)} + (v_k - v_j) \cdot e^{-t^B(v_j)}\\
&\geq \big(v_j - p^B(v_i)\big) \cdot e^{-t^B(v_i)} + (v_k - v_j) \cdot e^{-t^B(v_i)}\\
&= \big(v_k - p^B(v_i)\big) \cdot e^{-t^B(v_i)}.
\end{align*}
We can then apply induction to show that the constraints $\big(v_i - p^B(v_i)\big) \cdot e^{-t^B(v_i)} \geq \big(v_i - p^B(v_{j})\big) \cdot e^{-t^B(v_{j})}$ for every $i > j$ are satisfied.
\end{proof}

We can further rewrite the IC constraints as
\[
t(v_{i - 1}) - t(v_i) \geq \ln \frac{v_i - p(v_{i - 1})}{v_i - p(v_i)}, \quad \quad \forall i \in [n] \setminus \{1\}.
\]
Observe that $\{t(v_i)\}_{i \in [n]}$ only appear in the constraints of Program~\eqref{lp:2} and it is also subject to the constraint that $t(v_i) \in [0, T]$ for all $i \in [n]$. Moreover, since $t(v_i)$ is monotonically non-increasing as $i$ increases, it suffices to have $t(v_n) = 0$ and $t(v_1) \leq T$. We take the following perspective: given part of the solution $\{p(v_i)\}_{i = 1}^n$, in order to satisfy the constraints, we wish to minimize the total time span, i.e., $t(v_1)$. Observe that the total time span is minimized by setting:
\begin{equation}
\begin{cases}
t(v_n) = 0\\
t(v_{i - 1}) - t(v_i) =\ln \frac{v_i - p(v_{i - 1})}{v_i - p(v_i)}, \quad \quad \forall i \in [n] \setminus \{1\}
\end{cases}
\label{eq:time_relation}
\end{equation}
Therefore, the IC constraints can be further simplified as $t(v_1) = \sum_{i = 2}^n \ln \frac{v_i - p(v_{i - 1})}{v_i - p(v_i)} \leq T$.
%
%In fact, in the optimal solution of Program~\eqref{lp:2}, $t(v_1)$ should be exactly $T$; or otherwise, one can slightly increase $p(v_n)$%\footnote{For any $i \geq 2$, the optimal solution must have $p(v_i) < v_i$ to satisfy the IC constraints, otherwise $(v_i - p(v_i)) \cdot e^{-t(v_i)} = 0$, which cannot be at least $(v_i - p(v_{i - 1})) \cdot e^{-t(v_{i - 1})}$.}
%to get another feasible solution with better revenue. 
To summarize, we can now rewrite Program~\eqref{lp:2} as follows:
\begin{equation}
\tag{C}
\begin{array}{ll@{}ll}
\text{maximize}  & \displaystyle\sum\limits_{i = 1}^n p(v_i) f(v_i) &\\
\text{subject to}& \displaystyle\sum_{i = 2}^n \ln \frac{v_i - p(v_{i - 1})}{v_i - p(v_i)} \leq T, \\
&v_i - p(v_i) \geq 0, \quad \quad \forall i \in [n], \\
&p(v_i) - p(v_{i - 1}) \geq 0, \quad \quad \forall i \in [n] \setminus \{1\}. \\
\end{array}
\label{lp:3}
\end{equation}

\begin{proposition}
Any optimal solution of Program~(\ref{lp:3}) is an optimal solution of Program~(\ref{lp:2}).
\end{proposition}

We then show two useful lemmas regarding the structure of the optimal solutions of Program~(\ref{lp:3}). First, for the lowest value $v_1$, its price must be $p(v_1) = v_1$ in any optimal solution of Program~(\ref{lp:3}).

\begin{lemma}
\label{lem:pv1_v1}
In any optimal solution of Program~(\ref{lp:3}), $p(v_1) = v_1$.
\end{lemma}
\begin{proof}
First, the optimal solution satisfies $p(v_1) = \cdots = p(v_n) < v_1$, then we can simply set $p(v_1) = \cdots = p(v_n) = v_1$ to obtain a feasible solution that generates strictly higher revenue. 

Moreover, if the optimal solution satisfies $p(v_1) = p(v_2) = \cdots = p(v_k) < \min\{p(v_{k + 1}), v_1\}$ for some $k < n$, we create a pricing function such that $p'(v_i) = \min\{p(v_{k + 1}), v_1\}$ for $i \leq k$ and $p'(v_i) = p(v_i)$ for $i > k$. It is easy to verify that $\p'$ is a feasible solution that generates strictly higher revenue, which produces a contradiction.
% (or $p(v_1) = \cdots = p(v_n) < v_1$). We can set all of $p(v_1), p(v_2), \ldots, p(v_k)$ to $\min\big(p(v_{k + 1}), v_1\big)$ and increase the objective of Program~(\ref{lp:3}). Every constraint of Program~(\ref{lp:3}) still holds. This contradicts with the optimality premise.%$\sum_{i = 2}^n \ln \frac{v_i - p(v_{i - 1})}{v_i - p(v_i)}$ is decreased so the first constraint of Program~(\ref{lp:3}) still holds.
\end{proof}

Let $\p = \{p(v_i)\}_{i \in [n]}$ be the pricing vector. The next lemma demonstrates that the feasible set 
\[
    \S = \left\{\p: p(v_1) = v_1 \ \bigg| \ \sum_{i = 2}^n \ln \frac{v_i - p(v_{i - 1})}{v_i - p(v_i)} \leq T\right\}
\]
is convex, and moreover, the optimal solution of Program~\eqref{lp:3} is unique. 
\begin{lemma}
$\S$ is convex and the optimal solution of Program~\eqref{lp:3} is unique.
\label{lem:price_unique}
\end{lemma}
\begin{proof}
%Suppose $\mathbf{p}$ and $\mathbf{p'}$ are different optimal price vectors. Then $\mathbf{p''} := \frac{\mathbf{p} + \mathbf{p'}}{2}$ gives the same revenue. We will show it decreases $\sum_{i = 2}^n \ln \frac{v_i - p(v_{i - 1})}{v_i - p(v_i)}$:
Observe that, $\sum_{i = 2}^n \ln \frac{v_i - p(v_{i - 1})}{v_i - p(v_i)}$ can be written as
\begin{align*}
\sum_{i = 2}^n \ln \frac{v_i - p(v_{i - 1})}{v_i - p(v_i)} &= \ln \big(v_2 - p(v_1)\big) - \ln \big(v_n - p(v_n)\big) + \ln \prod_{i = 2}^{n - 1} \frac{v_{i + 1} - p(v_i)}{v_i - p(v_i)}\\
&= \ln \big(v_2 - p(v_1)\big) + \ln \frac{1}{v_n - p(v_n)} + \sum_{i = 2}^{n - 1} \ln \left(1 + \frac{v_{i + 1} - v_i}{v_i - p(v_i)}\right).
\end{align*}
Notice that $\ln \big(v_2 - p(v_1)\big)$ is a constant given $p(v_1) = v_1$, $\ln \frac{1}{v_n - p(v_n)}$ is convex in $p(v_n)$ and $\ln \left(1 + \frac{v_{i + 1} - v_i}{v_i - p(v_i)}\right)$ is convex in $p(v_i)$. As a result, if $\p \in \S$ and $\p' \in \S$,  we get $\p'' = \frac{\p + \p'}{2} \in \S$. Therefore, $\S$ is a convex set.

As for the uniqueness, we prove by contradiction. Suppose that there are two different optimal solutions $\mathbf{p}$ and $\mathbf{p'}$. Consider $\p'' = \frac{\p + \p'}{2}$, which is a feasible solution since $\S$ is convex and other constraints are linear. We get
\[
    \sum_{i = 1}^n p''_i(v_i) f(v_i) = \sum_{i = 1}^n p'_i(v_i) f(v_i) = \sum_{i = 1}^n p'_i(v_i) f(v_i).
\]
However, 
\[
    \sum_{i = 2}^n \ln \frac{v_i - p''(v_{i - 1})}{v_i - p''(v_i)} < \frac 1 2 \cdot \left(\sum_{i = 2}^n \ln \frac{v_i - p(v_{i - 1})}{v_i - p(v_i)} + \sum_{i = 2}^n \ln \frac{v_i - p'(v_{i - 1})}{v_i - p'(v_i)}\right) \leq T.
\]
We can now increase $p''(v_n)$ so that $\sum_{i = 2}^n \ln \frac{v_i - p''(v_{i - 1})}{v_i - p''(v_i)} = T$ to obtain another feasible solution with strictly greater revenue. This contradicts with the fact that $\mathbf{p}$ and $\mathbf{p'}$ are optimal solutions. Therefore, the optimal solution must be unique.
\end{proof}

\subsection{Grouping and an Exponential Time Warm-up Algorithm}
As a second step of developing a computationally efficient algorithm, we introduce the key concept of {\em grouping functions} and provide a warm-up algorithm that can compute the optimal pricing mechanism albeit in exponential time. Given a pricing mechanism with price $\p$, we can partition the values into groups such that values in the same group share the same price.
\begin{definition}\label{def:grouping}
Given $\p$, a \emph{grouping function} $g_{\p}  :  [n] \rightarrow [n]$ is a function that satisfies:
\begin{itemize}
    \item If $i = 1$ or $p(v_i) \neq p(v_{i - 1})$, then $g_{\p}(i) = i$.
    \item Otherwise, $g_{\p}(i) = g_{\p}(i - 1)$.
\end{itemize}
We simply write $g$ instead of $g_{\p}$ when $\p$ is clear from the context.
\end{definition}
Here, $v_i$ and $v_j$ are in the same group if $g(i) = g(j)$. Moreover, it is clear that $g$ is monotonically non-decreasing and we say a value $v_i$ is representative if $g(i) = i$.
Given a grouping function $g$, let $\I_g = \{i  :  g(i) = i\}$ be the set of indices of representative values in $g$. For convenience, let $\nxt_g(i) = \min\{j  :  g(v_j) > g(v_i)\}$ be the index of the representative value of the next group. For $\nxt_g(i)$, we will omit the dependence on $g$ when the context is clear. Moreover, let $f_g(k) = \sum_{i  :  g(i) = k} f(v_i)$ be the summation of the probability mass of values in the group with representative value $v_k$.

It turns out that given a grouping function $g$, by relaxing the monotonicity constraints on prices, one can compute the optimal price $\p$ respecting $g$. To be precise, consider the following Program~\eqref{lp:4}:
\begin{equation}
\tag{D}
\begin{array}{ll@{}ll}
\text{maximize}  & \displaystyle\sum\limits_{i = 1}^n p\left(v_{g(i)}\right) f(v_i) &\\
\text{subject to}& \displaystyle\sum_{i = 2}^n \ln \frac{v_i - p\left(v_{g(i-1)}\right)}{v_i - p\left(v_{g(i)}\right)} \leq T, \\
&v_i - p(v_i) \geq 0, \quad \quad \forall i \in \I_g \setminus \{1\}, \\
&p(v_1) = v_1.
\end{array}
\label{lp:4}
\end{equation}
In Program~\eqref{lp:4}, for $i \not \in \I_g$, its price $p(v_i)$ is set to be $p\left(v_{g(i)}\right)$ to respect the grouping function $g$. Moreover, $p(v_1)$ is set to be $v_1$ according to \cref{lem:pv1_v1}.

\begin{lemma}\label{lem:price_time}
    Given a grouping function $g$, there exists a constant $c > 0$ such that the optimal solution of Program~\eqref{lp:4} satisfies $p(v_1) = v_1$, $\sum_{i = 2}^n \ln \frac{v_i - p\left(v_{g(i-1)}\right)}{v_i - p\left(v_{g(i)}\right)} = T$, and
    \begin{itemize}
        \item For $k \in \I_g$ with $k > 1$ and $g(n) > k$, 
    \[
        p(v_k) = \frac{v_{\nxt(k)} + v_k - \sqrt{\left(v_{\nxt(k)} - v_k\right)^2 + 4 \cdot \left(v_{\nxt(k)} - v_k\right) / \left(c \cdot f_g(k)\right)}} 2;
    \]
        \item For $k > 1$ with $g(n) = k$, $p(v_k) = v_k - 1 / \left(c \cdot f_g(k)\right)$.
    \end{itemize}
\end{lemma}
\begin{proof}
    Notice that the constraints of Program~\eqref{lp:4} are convex, the objective is linear, and Slater's condition clearly holds (by considering $p(v_k) = 0$ for all $k \in \I_g$). As a result, Theorem 1 in Chapter 8.6 of \citet{luenberger1997optimization} implies that the strong duality holds and the Lagrangian objective is
    \begin{align*}
        \L(\p; c, \alpha) &~= \sum\limits_{i = 1}^n p\left(v_{g(i)}\right) f(v_i) - \frac 1 c \cdot \left(\sum_{i = 2}^n \ln \frac{v_i - p\left(v_{g(i-1)}\right)}{v_i - p\left(v_{g(i)}\right)} - T\right) + \sum_{i  :  g(i) = i} \alpha(i) \cdot \big(v_i - p(v_i)\big).
    \end{align*}
    We consider the first order conditions with respect to $p(v_k)$ each $k \in \I_g \setminus \{1\}$. 
    % For $k = 1$, we have
    %     \[
    %         \frac{\partial \L(\p; c, \alpha)}{\partial p(v_1)} = f_g(1) + \frac 1 c \cdot \frac{1}{v_{\nxt(1)} - p(v_1)}- \alpha(1) = 0,
    %     \]
    % which implies that $\alpha(1) > 0$. Thus, by complementary slackness, we must have $p(v_1) = v_1$. 
    For $k \in \I_g$ with $k > 1$ and $g(n) > k$, we have
        \[
            \frac{\partial \L(\p; c, \alpha)}{\partial p(v_k)} = f_g(k) + \frac 1 c \cdot \left(\frac{1}{v_{\nxt(k)} - p(v_k)} - \frac{1}{v_{k} - p(v_k)}\right) - \alpha(k) = 0,
        \]
    Notice that if $\alpha(k) > 0$, then by complementary slackness, we have $v_k - p(v_k) = 0$, which implies that $\frac{\partial \L(\p; c, \alpha)}{\partial p(v_k)} = -\infty$. As a result, $\alpha(k)$ must be $0$, and we can now compute $p(v_k)$ by rearranging the terms in $\frac{\partial \L(\p; c, \alpha)}{\partial p(v_k)}$.
    Finally, for $k > 1$ with $g(n) = k$, we have
    \[
        \frac{\partial \L(\p; c, \alpha)}{\partial p(v_k)} = f_g(k) - \frac 1 c \cdot \frac{1}{v_{k} - p(v_k)} - \alpha(k) = 0,
    \]
    Similar to the previous case, $\alpha(k) = 0$, and therefore, we can compute $p(v_k)$ by rearranging the terms.
\end{proof}

Given a grouping function $g$, \cref{lem:price_time} characterizes the optimal solution for Program~\eqref{lp:4} by a constant $c > 0$. The next lemma enables an efficient way to search for $c$. 
\begin{lemma}
\label{lem:time_monotonicity_c}
$\sum_{i = 2}^n \ln \frac{v_i - p(v_{i - 1})}{v_i - p(v_i)}$ is monotonically increasing as $c$ increases.
\end{lemma}
\begin{proof}
Observe that $\sum_{i = 2}^n \ln \frac{v_i - p(v_{i - 1})}{v_i - p(v_i)}$ can be written as:
\begin{align*}
\sum_{i = 2}^n \ln \frac{v_i - p(v_{i - 1})}{v_i - p(v_i)} &= \ln \prod_{i = 2}^n \frac{v_i - p(v_{i - 1})}{v_i - p(v_i)}\\
&= \ln \left(\frac{v_2 - p(v_1)}{v_n - p(v_n)} \cdot \prod_{i = 2}^{n - 1}\left(1 + \frac{v_{i + 1} - v_i}{v_i - p(v_i)}\right)\right)\\
&= \ln \left(\frac{v_2 - v_1}{v_n - p(v_n)} \cdot \prod_{i = 2}^{n - 1}\left(1 + \frac{v_{i + 1} - v_i}{v_i - p(v_i)}\right)\right),
\end{align*}
where the last equality follows from $p(v_1) = v_1$.
Therefore, $\sum_{i = 2}^n \ln \frac{v_i - p(v_{i - 1})}{v_i - p(v_i)}$ is monotonically increasing as $p(v_i)$ increases. Finally, notice that $p(v_i)$ is monotonically increasing as $c$ increases, and therefore, we conclude that $\sum_{i = 2}^n \ln \frac{v_i - p(v_{i - 1})}{v_i - p(v_i)}$ monotonically increases as $c$ increases.
\end{proof}

\cref{lem:time_monotonicity_c} warrants a binary search approach to find the desired constant $c$. Combining with our structural result of \cref{lem:price_time}, we are now ready to develop a simple algorithm to compute the optimal pricing mechanism by enumerating all possible grouping function (see \cref{alg:enumerate_grouping}).

\begin{algorithm}
	\DontPrintSemicolon
	\ForEach{$v^{\mathrm{min}}$ being the smallest value participating in the auction} {
	    Rename $v_1 < v_2 < \cdots < v_n$ to be the values participating in the auction\;
	    \ForEach{possible grouping function $g$} {
	        Solve the optimal solution for Program~\eqref{lp:4} corresponding to $g$\;
	        \If{the prices in the optimal solution are monotonically non-decreasing}{Calculate the revenue and record the solution\;}
        }
	}
	\Return the best recorded solution\;
	\caption{Warm-up (exponential time) algorithm by grouping enumeration}\label{alg:enumerate_grouping}
\end{algorithm}

\definecolor{inv}{rgb}{0.70, 0.70, 0.70}
\definecolor{win}{rgb}{0.5, 0.95, 0.5}
\begin{table}[ht]
\centering
 \begin{tabular}{| c | c | c | c | c || c|} 
 \hline
 $v^{\mathrm{min}}$ & $\big(g(1), g(2), g(3)\big)$ & $p(v_1)$ & $p(v_2)$ & $p(v_3)$ & $\rev$ \\ [0.5ex] 
 \hline\hline
 $v_1$ & $(1, 2, 3)$ & $3$ & $2$ & $9.5$ & \cellcolor{inv}N.A. ($4.833$) \\ 
 \hline
 \cellcolor{win}$v_1$ & \cellcolor{win}$(1, 1, 3)$ & \cellcolor{win}$3$ & \cellcolor{win}$3$ & \cellcolor{win}$7.5$ & \cellcolor{win}$4.5$ \\ 
 \hline
 $v_1$ & $(1, 2, 2)$ & $3$ & $3.5$ & $3.5$ & $3.333$ \\ 
 \hline
 $v_1$ & $(1, 1, 1)$ & $3$ & $3$ & $3$ & $3$ \\ 
 \hline
 $v_2$ & $(\cdot, 2, 3)$ & \cellcolor{inv}$\infty$ & $4$ & $8$ & $4$ \\ 
 \hline
 $v_2$ & $(\cdot, 2, 2)$ & \cellcolor{inv}$\infty$ & $4$ & $4$ & $2.667$ \\ 
 \hline
 $v_3$ & $(\cdot, \cdot, 3)$ & \cellcolor{inv}$\infty$ & \cellcolor{inv}$\infty$ & $12$ & $4$ \\ 
 \hline
\end{tabular}
\caption{Results of all iterations when running \cref{alg:enumerate_grouping} on \cref{ex:enumerate_grouping}}
\label{tab:enumerate_grouping}
\end{table}
\begin{example}
Let $n = 3$. $(v_1, v_2, v_3) = (3, 4, 12)$, and $(f(v_1), f(v_2), f(v_3)) = \big(\frac{1}{3}, \frac{1}{3}, \frac{1}{3}\big)$. Let $T = \ln 2$. \cref{alg:enumerate_grouping} enumerates the minimum value $v^{\mathrm{min}}$ and the grouping function $g$. The results for all possible combinations are shown in \cref{tab:enumerate_grouping}. In particular, the first row is invalid since $p(v_2) < p(v_1)$; while the second row (highlighted) corresponds to the optimal solution.
% In particular, for the iteration with $v^{\mathrm{min}} = v_1$ and $(g(v_1), g(v_2), g(v_3)) = (1, 2, 3)$, \cref{alg:enumerate_grouping} calculates the prices using \cref{lem:price_time}, leading $p(v_1) = v_1 = 3$, $p(v_2) = \frac{v_3 + v_2 - \sqrt{(v_3 - v_2)^2 + 4 \cdot \frac{v_3 - v_2}{c \cdot f(v_2)}}}{2} = 8 - \sqrt{16 + \frac{24}{c}}$, and $p(v_3) = v_3 - \frac{1}{c \cdot f(v_3)} = 12 - \frac{3}{c}$. We have
% \[
% \ln \frac{v_3 - p(v_2)}{v_3 - p(v_3)} + \ln \frac{v_2 - p(v_1)}{v_2 - p(v_2)} = \ln \frac{4 + \sqrt{16 + \frac{24}{c}}}{\frac{3}{c}} + \ln \frac{1}{\sqrt{16 + \frac{24}{c}} - 4} = \ln \frac{\left(2 + \sqrt{4 + \frac{6}{c}}\right)^2c^2}{18}.
% \]
% When we set it to the horizon $T = \ln 2$, we derive $c = 1.2$, and thus $p(v_1) = 3$, $p(v_2) = 2$ and $p(v_3) = 9.5$. However, this set of prices is invalid since $p(v_1) > p(v_2)$.
% Now we focus on another iteration where $v^{\mathrm{min}} = v_1$ and $(g(v_1), g(v_2), g(v_3)) = (1, 1, 3)$. Since $v_1$ and $v_2$ are in the same group, we put them together and create a new instance, where $n' = 2$, $v'_1 = v_1 = 3$, $v'_2 = v_3 = 12$, $f'(v'_1) = f(v_1) + f(v_2) = \frac{2}{3}$ and $f'(v'_2) = f(v_3) = \frac{1}{3}$. The solved prices turn out to be $p'(v'_1) = 3$ and $p'(v'_2) = 7.5$, which is monotone. It gives revenue of $4.5$, and this is optimal across all iterations as demonstrated in \cref{tab:enumerate_grouping}.
\label{ex:enumerate_grouping}
\end{example}

\iffalse
\begin{proposition}
Every $t(v_i)$ and $p(v_i)$ is continuous in $c \in (0, +\infty)$.
\end{proposition}
\fi
% \begin{theorem}
% \cref{alg:enumerate_grouping} outputs an optimal solution.
% \end{theorem}
% \begin{proof}
% The succinct grouping $g^*$ of the optimal solution is visited during the enumeration. At that iteration, \cref{alg:enumerate_grouping} can find the optimal prices using \cref{lem:price_time}. \cref{lem:time_monotonicity_c} enables finding the optimal $c^*$ with binary search.
% \end{proof}

\subsection{Computationally Efficient Algorithm}
\label{subsec:opt_alg}
In this section, we improve \cref{alg:enumerate_grouping} to develop a computationally efficient algorithm for calculating the optimal pricing mechanism. Instead of solving the optimal pricing mechanism for a particular $T$, we aim to solve it for all possible $T' \geq T$ all at once. 
% For this purpose, we introduce the concept of succinct grouping functions.
% \begin{definition}[Succinct Grouping Function]
% Given a pricing vector $\p$, the \emph{succinct grouping function} $g^*_{\p} : [n] \rightarrow [n]$ is the grouping function that induces minimum number of groups, and therefore, it is a grouping function that satisfies:
% \begin{itemize}
%     \item If $i > 1$ and $p(v_i) = p(v_{i - 1})$, then $g^*_{\p}(i) = g^*_{\p}(i - 1)$.
% \end{itemize}
% \end{definition}
Intuitively, our algorithm starts from time limit $T' = \infty$ and gradually decreases $T'$ until it hits $T$. For each time limit $T'$, the algorithm maintains the grouping function corresponding to the optimal pricing function.

Inspired by \cref{lem:price_time}, given a grouping function $g$, let $p_g(\cdot; c)$ be a pricing function such that
    \begin{itemize}
        \item $p_g(v_1; c) = v_1$,
        \item For $k \in \I_g$ with $k > 1$ and $g(n) > k$, 
    \[
        p_g(v_k; c) = \frac{v_{\nxt(k)} + v_k - \sqrt{\left(v_{\nxt(k)} - v_k\right)^2 + 4 \cdot \left(v_{\nxt(k)} - v_k\right) / \left(c \cdot f_g(k)\right)}} 2;
    \]
        \item For $k > 1$ with $g(n) = k$, $p_g(v_k; c) = v_k - 1 / \left(c \cdot f_g(k)\right)$.
    \end{itemize}
%For $i \not \in \I_g$, simply set $p(v_i) = p(v_{g(i)})$. 
Moreover, for $k \in \I_g$ with $g(n) > k$, let 
\[
    c_g^*(k) = \min\{c > 0  :  p_g(v_k; c) \leq p_g(v_{\nxt(k)}; c)\}
\]
be the minimum $c$ such that the price for $v_k$ is still not larger than the price for $v_{\nxt(k)}$. Finally, let $c_g^*(k) = 0$ for $k = g(n)$. We are now ready to present our computationally efficient algorithm (see \cref{alg:poly_time}). Intuitively, when $T'$ approaches $+\infty$, the optimal pricing function $p(v_i)$ approaches $v_i$, and the seller can almost extract the full welfare $\sum_{i \in [n]} v_i f(v_i)$ as the revenue. Therefore, when starting from a large enough $T'$, each value $v_i$ forms a separate group initially. For each step, under the current grouping function $g$, we compute the maximum time limit $T'$ such that there exists two values from different groups sharing the same prices in the optimal solution. We merge these two groups and repeat the process, until $T' \leq T$. In the end, we compute the optimal solution for time limit $T$ under the final grouping function.

\begin{algorithm}
	\DontPrintSemicolon
	\ForEach{$v^{\mathrm{min}}$ being the smallest value participating in the auction} {
	   % Create a temporary copy of the instance to work on in this iteration\;
	   Rename $v_1 < v_2 < \cdots < v_n$ to be the values participating in the auction\;
	   % %$g(i) \gets i, \quad \forall i$ \tcp*{Let each value be in a separate group in the beginning.}
	   Initialize the grouping function $g$ with $g(i) = i$ for all $i \in [n]$\;
	   %\;
	   $T' \gets \infty$\;
	    \Repeat{$T' \leq T$}
	    {
	        Let $k^* \in \arg\max_{k} c_g^*(k)$\;
	        $T' \gets \sum_{i = 2}^n \ln \left(v_i - p_g\big(v_{g(i-1)};~ c_g^*(k^*)\big)\right) - \ln \left(v_i - p_g\big(v_{g(i)};~c_g^*(k^*)\big)\right)$\;
    	    \If{$T' > T$} {
    	        \ForEach{$i  :  g(i) = k^*+1$}
    	        {
    	            $g(i) \gets k^*$ \tcp*{Combine groups with index $k^*$ and $(k^* + 1)$}
    	        }
        	}
    	}
    	Solve the optimal solution for Program~\eqref{lp:4} corresponding to $g$\;

    	Calculate the revenue and record the solution\;
	}
	\Return the best recorded solution\;
	\caption{Computationally efficient algorithm for the optimal pricing mechanism}\label{alg:poly_time}
\end{algorithm}

Observe that in \cref{alg:poly_time}, there are $|V|$ possible $v^{\mathrm{min}}$ to enumerate. Moreover, there are $|V|$ groups initially, and for each repeat-until loop, the number of groups decreases by $1$. Therefore, \cref{alg:poly_time} is computationally efficient. %in time $O(|V|^2)$, and thus, it is computationally efficient.
In the rest of this subsection, we show the correctness of \cref{alg:poly_time}.
We begin with the following lemma.
\begin{lemma}\label{lem:price_monotonicity_c}
For any grouping function $g$, and for any $k \in \I_g$ with $k < g(n)$, if $c < c_g^*(k)$, then we have $p_g(v_k; c) > p_g(v_{\nxt(k)}; c)$.
\end{lemma}
\begin{proof}
For $k \in \I_g$ with $k > 1$ and $k < g(n)$, we take the derivative of $p_g(v_k; c)$ with respect to $c$,
% \begin{align*}
% \frac{\partial p_g(v_k; c)}{\partial c} %&= -\frac{1}{2} \cdot \frac{\partial \sqrt{\left(v_{\nxt(k)} - v_k\right)^2 + 4 \cdot \left(v_{\nxt(k)} - v_k\right) / \left(c \cdot f_g(k)\right)}}{\partial c} \\
% = \frac{v_{\nxt(k)} - v_k}{c^2 \cdot f_g(k) \cdot \sqrt{\left(v_{\nxt(k)} - v_k\right)^2 + 4 \cdot \left(v_{\nxt(k)} - v_k\right) / \left(c \cdot f_g(k)\right)}}.
% \end{align*}
% Plug in the formula of $p_g(v_k; c)$ 
and with a few steps of algebraic manipulations, we have
\begin{align*}
\frac{\partial p_g(v_k; c)}{\partial c} 
% &= 
% \frac{v_{\nxt(k)} - v_k}{c^2 \cdot f_g(k) \cdot \big(v_{\nxt(k)} + v_k - 2 \cdot p_g(v_k; c)\big)}\\
% &= \frac{(v_{i + 1} + v_i - 2p(v_i))^2 - (v_{i + 1} - v_i)^2}{4c (v_{i + 1} + v_i - 2p(v_i))}\\
%&= \frac{4p^2(v_i) + 4v_{i + 1}v_i - 4p(v_i)(v_{i + 1} + v_i)}{4c (v_{i + 1} + v_i - 2p(v_i))}\\
% &= \frac{(v_{i + 1} - p(v_i))(v_i - p(v_i))}{c (v_{i + 1} + v_i - 2p(v_i))}\\
&= \frac{1}{c} \cdot \left(\frac{1}{v_{\nxt(k)} - p_g(v_k; c)} + \frac{1}{v_k - p_g(v_k; c)}\right)^{-1}.
\end{align*}
As for $k = g(n)$, similarly we have
\[
\frac{\partial p_g(v_k; c)}{\partial c} = \frac{1}{c} \cdot \left(\frac{1}{v_n - p_g(v_k; c)}\right)^{-1}.
\]
Finally, $\frac{\partial p_g(v_1; c)}{\partial c} = 0$. Therefore, for any $k \in \I_g$ with $k < g(n)$, when $p_g(v_k; c) = p_g(v_{\nxt(k)}; c)$ for $c = c_g^*(k)$, we have $\frac{\partial p_g(v_{\nxt(k)}; c)}{\partial c} > \frac{\partial p_g(v_k; c)}{\partial c}$, which concludes the proof.
\end{proof}

Intuitively, \cref{lem:price_monotonicity_c} demonstrates that for any grouping $g$, its pricing function $p_g(\cdot; c)$ violates the monotonicity constraints for any $c < \max_k c_g^*(k)$. Recall that by \cref{lem:time_monotonicity_c}, given a grouping function $g$, there is in fact a bijection between $c$ and the time limit $T'$. We say a grouping function $g$ is valid with time limit $T'$ if 
\[
    T^*(g) = \sum_{i = 2}^n \ln \left(v_i - p_g\big(v_{g(i-1)};~ \max_k c_g^*(k)\big)\right) - \ln \left(v_i - p_g\big(v_{g(i)};~\max_k c_g^*(k)\big)\right) \leq T',
\]
and therefore, $T^*(g)$ is the minimum time limit in which $g$ is valid. Moreover, for each grouping function $g$, denote the optimal revenue under grouping function $g$ when time limit $T' \geq T^*(g)$ by
\[ 
    \rev_g(T') = \sum_{i = 1}^n p_g\left(v_{g(i)}; c_g(T')\right) f(v_i)
\]
where $c_g(T')$ satisfies 
\[
    \sum_{i = 2}^n \ln \left(v_i - p_g\big(v_{g(i-1)};~ c_g(T')\big)\right) - \ln \left(v_i - p_g\big(v_{g(i)};~c_g(T')\big)\right) = T'.
\]
For $T' < T^*(g)$, simply set $\rev_g(T') = -\infty$. As a result, for each time limit $T'$, the optimal revenue is simply given by 
\[
    \rev^*(T') = \max_g \rev_g(T').
\]
We are now ready to prove the following key lemma that leads to the correctness of \cref{alg:poly_time}.

\begin{lemma}\label{lem:key-correctness}
    If $g^* \in \arg\max_g \rev_g(T')$ for some $T' > T^*(g^*)$, then for all $T''$ satisfying $T^*(g^*) \leq T'' < T'$, we have $g^* \in \arg\max_g \rev_g(T'')$.
\end{lemma}
\begin{proof}
    For the purpose of contradiction, assume that there exists $T^*(g^*) \leq T'' < T'$ such that $g^* \not \in \arg\max_g \rev_g(T'')$. It implies that there exists a grouping function $g'$ satisfying $\rev_{g'}(T'') > \rev_{g^*}(T'')$. Note that since for any $g$, $\rev_g(\tilde T)$ is continuous in $\tilde T$ whenever $\tilde T \geq T^*(g)$, there must exist a time $\bar T$ satisfying $T'' < \bar T \leq T'$ and $\rev_{g'}(\bar T) = \rev_{g^*}(\bar T)$.
    
    By \cref{lem:price_unique}, the optimal pricing function when time limit is $\bar T$ is unique, and we denote it by $\bar \p$. As a result, $g'$ is a grouping function that is consistent with $\bar \p$ such that $g'(i) \neq g'(j)$ if $\bar p(v_i) \neq \bar p(v_j)$. However, if there exists $i \neq j$ such that $\bar p(v_i) = \bar p(v_j)$ but $g'(i) \neq g'(j)$, then it implies that $T^*(g') \geq \bar T$ by \cref{lem:price_monotonicity_c}. As a result, the unique grouping function that is consistent with $\bar \p$ and also valid at $T''$ is $g_{\bar \p}$, constructed according to \cref{def:grouping}. This implies $g^* = g_{\bar \p}$ and contradicts with the fact that $\rev_{g'}(T'') > \rev_{g^*}(T'')$.
\end{proof}

\cref{lem:key-correctness} demonstrates that once we successfully identify the grouping function corresponding to the optimal pricing function for time limit $T'$, then such a grouping function will continue to be the grouping function corresponding to the optimal pricing function for time limit $T'' < T'$ until it becomes invalid. This is exactly how \cref{alg:poly_time} proceeds. The algorithm starts with time limit $T' = \infty$ in which the optimal pricing function is simply $p(v_i) = v_i$ and therefore, its corresponding grouping function is $g(i) = i$ for all $i \in [n]$. In each iteration of the repeat-until loop, for the current grouping function $g$, it essentially computes $T^*(g)$. At $T^*(g)$, the old grouping function $g$ is equivalent to a new, coarser grouping function $g'$ which groups the values sharing the same prices together, and both of $g$ and $g'$ give the same, optimal prices. $g$ is about to become invalid if we keep decreasing $T'$, and therefore \cref{alg:poly_time} replaces the current grouping by $g'$. We can then apply \cref{lem:key-correctness} again and know that $g'$ gives optimal prices until it becomes invalid. The next theorem formalized the argument above.

\begin{theorem}\label{thm:poly_time_optimal}
\cref{alg:poly_time} outputs the optimal pricing mechanism for time limit $T$.
\end{theorem}

\paragraph{Continuous value distributions.} \cref{alg:poly_time} takes as input a discrete value distribution. For a continuous value distribution supported on $[0, M]$, we can discretize the support, and run \cref{alg:poly_time} on the discretized distribution. This gives a fully polynomial-time approximation scheme (FPTAS) to the revenue maximization problem if we have oracle access to (inverse) CDF of the value distribution. We discuss this in detail in \cref{sec:continuous}.

\paragraph{Grouping in the middle.} Here we provide an example showing that grouping can happen in the middle of the supports and the monotonicity constraints can be binding there, as opposed to the behaviors in \citep{Shneyerov14,CEP19,wang2001optimal}. This demonstrates that general value distributions indeed pose additional challenges than restrictive ones in previous work.
\begin{example}
\label{ex:grouping_necessary}
Let $n = 4$, $(v_1, v_2, v_3, v_4) = (100, 101, 102, 103)$, and $\big(f(v_1), f(v_2), f(v_3), f(v_4)\big) = \big(\frac{1}{3} - \varepsilon, \frac{1}{3}, \varepsilon, \frac{1}{3}\big)$ for a small enough $\varepsilon > 0$. Clearly the optimal pricing curve sells to all values, i.e., $v^{\mathrm{min}} = v_1$, since any pricing curve with $v^{\mathrm{min}} \geq v_2$ cannot generate revenue more than $v_4 \cdot (1 - f(v_1)) = \frac{206}{3}$, which is less than that of simply pricing at $v_1$. Let $g(\cdot)$ be the grouping function that puts every value into a separate group, i.e., $g(i) = i$ for $i = 1, 2, 3, 4$. We calculate $p_g(v_i; c)$'s according to \cref{alg:poly_time}:
\begin{itemize}
    \item $p_g(v_1; c) = v_1 = 100$,
    \item $p_g(v_2; c) = \frac{2v_2 + 1 - \sqrt{1 + 4 / \left(c \cdot f_g(2)\right)}} 2 = 101 + \frac{1 - \sqrt{1 + 12 / c}} 2$,
    \item $p_g(v_3; c) = \frac{2v_3 + 1 - \sqrt{1 + 4 / \left(c \cdot f_g(3)\right)}} 2 = 102 + \frac{1 - \sqrt{1 + 4 / \left(c \cdot \varepsilon\right)}} 2$,
    \item $p_g(v_4; c) = v_4 - 1 / \left(c \cdot f_g(4)\right) = 103 - 3 / c$.
\end{itemize}
As $c$ goes down from $+\infty$, the first monotonicity constraint to fail is $p_g(v_2; c) \leq p_g(v_3; c)$, at $c \approx 1 / (2 \varepsilon)$. Therefore, the first merge in \cref{alg:poly_time} is on $v_2$ and $v_3$. For a $T$ slightly less than the merging point, we have $p(v_1) < p(v_2) = p(v_3) < p(v_4)$.
\end{example}

\subsection{Implications on Uniform Distributions}
\label{subsec:uniform}
In this section, we demonstrate the power of \cref{alg:poly_time} by deriving the optimal pricing curve when the value distribution is uniform; and its optimal pricing curve turns out to have a nice structure. We focus on the uniform distribution $U[0, 1]$. Using our discretization technique (discussed in \cref{sec:continuous}), we consider the discrete uniform distribution over $V = \{\varepsilon, 2 \cdot \varepsilon, \ldots, \frac{1}{\varepsilon} \cdot \varepsilon\}$, i.e., $f(i \cdot \varepsilon) = \varepsilon$, where $\varepsilon \to 0^+$.

Let $x = v_1$ be the minimum value $v^{\mathrm{min}}$ to purchase in the mechanism. Using the characterization in \cref{lem:price_time}, if a value $v_i$ with $x < v_i < 1$ forms a group on its own, then
\[
p(v_i) = \frac{v_{i + 1} + v_i - \sqrt{(v_{i + 1} - v_i)^2 + 4 \cdot \frac{v_{i + 1} - v_i}{c \cdot f(v_i)}}}{2} = v_i - \frac{-\varepsilon + \sqrt{\varepsilon^2 + \frac{4}{c}}}{2}.
\]
As a result, if $v_{i+1} < 1$ also forms a group on its own, then we simply have $p(v_{i+1}) - p(v_i) = \varepsilon$. Further, fixing $x < v_i < v_{i + 1} < 1$, if $v_i$ and $v_{i + 1}$ each forms a group on its own, then throughout \cref{alg:poly_time}, they will never be merged together, since $p(v_{i+1}) \neq p(v_i)$. Therefore, merging can only happen at the beginning and the end of the value spectrum, and thus the grouping function $g$ must satisfy $g(i) \in \{1, i, g(n)\}$. ($g(i)$ is $1$ if it was merged to the first group, is $g(n)$ if it was merged to the last group, and is $i$ if it was never merged.) Thus, the optimal pricing function $p(v)$ must have three thresholds $0< x < y < z < 1$ such that
\[
p(v) = 
\begin{cases}
\infty &\qquad \text{if } v < x \\
x &\qquad \text{if } x \leq v < y\\
v - y + x &\qquad \text{if } y \leq v < z\\
z - y + x &\qquad \text{if } z \leq v \leq 1
\end{cases}.
\]
In other words, the pricing function is flat when $v \in [x, y) \cup [z, 1]$ and it has slope $1$ when $v \in [y, z)$. Therefore, given this pricing function, the total time span is about:
\[
\sum_{i = y/\varepsilon}^{z/\varepsilon} \ln\frac{y - x + \varepsilon}{y - x} = \frac{z - y}{y - x} + O(\varepsilon),
\]
and as $\varepsilon \to 0^+$, the total revenue is
\[
(y - x) x + (z - y)\left(\frac{y + z}{2} - y + x\right) + (1 - z)(z - y + x) = x - x^2 - y + \frac{y^2}{2} + z - \frac{z^2}{2}.
\]

Optimizing it under the constraint $(z - y) / (y - x) = T$ where $T$ is the time limit, we have $x = \frac 2 {T+4}$, $y = \frac 3 {T+4}$, and $z = \frac {T+2} {T+4}$ with optimal revenue $\frac{T + 2}{2T + 8}$. Therefore, the optimal pricing function is:
\[
p(v) = 
\begin{cases}
\infty &\qquad \text{if } v < \frac{2}{T + 4} \\
\frac{2}{T + 4} &\qquad \text{if } \frac{2}{T + 4} \leq v < \frac{3}{T + 4}\\
v - \frac{1}{T + 4} &\qquad \text{if } \frac{3}{T + 4} \leq v < \frac{T + 3}{T + 4}\\
\frac{T + 2}{T + 4} &\qquad \text{if } \frac{T + 3}{T + 4} \leq v \leq 1
\end{cases}.
\]

\cref{fig:pricing_curves} plots $p(v)$ as a function of $v$ and $t$ for different time limit $T$.
As we can see in \cref{fig:p_vs_t}, the optimal pricing curve is linear in time $t$ for any time limit $T$. Combining with \cref{fig:p_vs_v}, we can observe that there are four different purchasing behaviors, depending on the buyer's value. In particular, 
\begin{itemize}
    \item if the buyer has a high value $v \in [\frac{T + 3}{T + 4}, 1]$, she will purchase at $t = 0$;
    \item if the buyer has a medium-high value $v \in [\frac{3}{T + 4}, \frac{T + 3}{T + 4})$, she will purchase at a time $t \in (0, T)$;
    \item if the buyer has a medium-low value $v \in [\frac{2}{T + 4}, \frac{3}{T + 4})$, she will purchase at a time $t = T$;
    \item if the buyer has a low value $v \in [0, \frac{2}{T + 4})$, she will not participate into the mechanism.
\end{itemize} 
% =========== figure ===============

\definecolor{c0}{rgb}{0.9,0.1,0}
\definecolor{c1}{rgb}{0.6,0.1,0.6}
\definecolor{c6}{rgb}{0.2,0.3,0.8}
\definecolor{cinf}{rgb}{0.1,0.7,0.4}
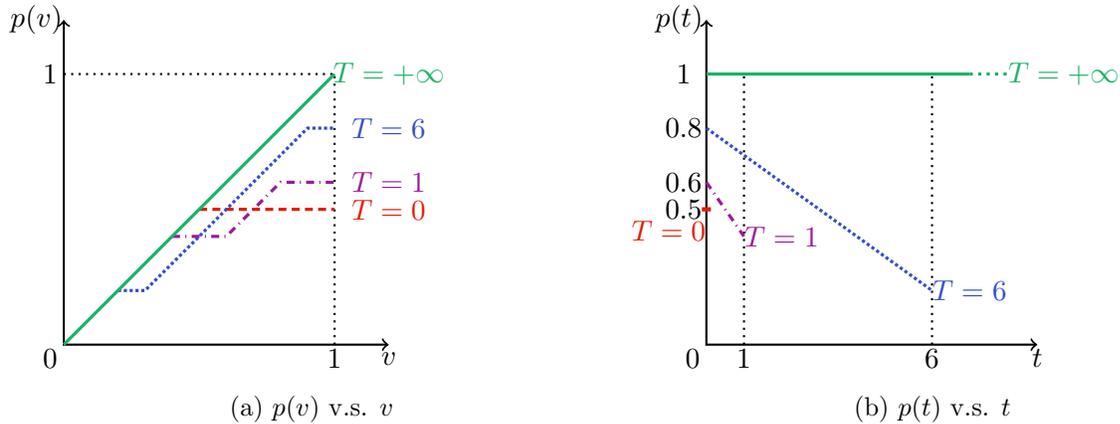
\begin{figure}[t]
\centering
\begin{subfigure}{.5\textwidth}
\begin{tikzpicture}[scale = 3.6]
\draw [thick, <->] (0,1.2) -- (0,0) -- (1.2,0);

\node at (1.2, -0.05) {$v$};
\node at (1.0, -0.05) {$1$};
\node at (-0.1, 1.2) {$p(v)$};
\node at (-0.05, 1.0) {$1$};
\node at (-0.05, -0.05) {$0$};

\node [c0] at (1.2, 0.5) {$T = 0$};
\node [c1] at (1.2, 0.6) {$T = 1$};
\node [c6] at (1.2, 0.8) {$T = 6$};
\node [cinf] at (1.2, 1.0) {$T = +\infty$};

\draw [thick, dotted] (1, 0) -- (1, 1);
\draw [thick, dotted] (0, 1) -- (1, 1);

\draw[c0, densely dashed, very thick] (0.5, 0.5) -- (1.0, 0.5);

\draw[c1, dashdotted, very thick] (0.4, 0.4) -- (0.6, 0.4) -- (0.8, 0.6) -- (1, 0.6);

\draw[c6, densely dotted, very thick] (0.2, 0.2) -- (0.3, 0.2) -- (0.9, 0.8) -- (1, 0.8);

\draw[cinf, very thick] (0.0, 0.0) -- (1.0, 1.0);

\end{tikzpicture}
\caption{$p(v)$ v.s. $v$}
\label{fig:p_vs_v}
\end{subfigure}%
\begin{subfigure}{.5\textwidth}
\begin{tikzpicture}[scale = 0.5, yscale = 7.2]
\draw [thick, <->] (0,1.2) -- (0,0) -- (8.8,0);

\node at (8.8, -0.05) {$t$};
\node at (1, -0.05) {$1$};
\node at (6, -0.05) {$6$};
\node at (-0.72, 1.2) {$p(t)$};
\node at (-0.6, 1.0) {$1$};
\node at (-0.6, 0.8) {$0.8$};
\node at (-0.6, 0.6) {$0.6$};
\node at (-0.6, 0.5) {$0.5$};
\node at (-0.36, -0.05) {$0$};

\node [c0] at (-1.0, 0.42) {$T = 0$};
\node [c1] at (2.0, 0.4) {$T = 1$};
\node [c6] at (7.0, 0.2) {$T = 6$};
\node [cinf] at (9.5, 1.0) {$T = +\infty$};

\draw [thick, dotted] (1, 0) -- (1, 1);
\draw [thick, dotted] (6, 0) -- (6, 1);

\draw[domain=-0.12:0.12, variable=\x, c0, samples=200, smooth, ultra thick] plot ({\x}, {0.5});
\draw[domain=0:1, variable=\x, c1, dashdotted, samples=200, smooth, very thick] plot ({\x}, {0.6-0.2*\x});
\draw[domain=0:6, variable=\x, c6, densely dotted, samples=200, smooth, very thick] plot ({\x}, {0.8-0.1*\x});
\draw[domain=0:7, variable=\x, cinf, samples=200, smooth, very thick] plot ({\x}, {1.0});
\draw[domain=7:8, variable=\x, cinf, samples=200, smooth, very thick, dotted] plot ({\x}, {1.0});
\end{tikzpicture}
\caption{$p(t)$ v.s. $t$}
\label{fig:p_vs_t}
\end{subfigure}
\caption{Pricing curves for uniform distribution $U[0, 1]$}
\label{fig:pricing_curves}
\end{figure}

\begin{figure}[t]
\centering
\begin{tikzpicture}[scale = 0.5, yscale = 14.4]
\draw [thick, <->] (0,0.6) -- (0,0) -- (12.8,0);

\node at (12.8, -0.03) {$T$};
\node at (-0.6, -0.025) {$0$};
\node at (1, -0.025) {$1$};
\node at (6, -0.025) {$6$};
\node at (-1.5, 0.6) {Revenue};
\node at (-0.6, 0.5) {$0.5$};
\node at (-0.9, 0.25) {$0.25$};

\node [c0] at (-1.0, 0.20) {$T = 0$};
\node [c1] at (2.0, 0.26) {$T = 1$};
\node [c6] at (7.0, 0.36) {$T = 6$};
\node [cinf] at (13.2, 0.5) {$T = +\infty$};

\draw [thick, dotted] (1, 0) -- (1, 0.5);
\draw [thick, dotted] (6, 0) -- (6, 0.5);
\draw [thick, dotted, cinf] (0, 0.5) -- (12, 0.5);

\draw[domain=0:12, variable=\x, samples=2000, smooth, very thick] plot ({\x}, {(\x + 2)/(2 * \x + 8)});

\draw[domain=-0.12:0.12, variable=\x, c0, samples=200, smooth, ultra thick] plot ({\x}, {0.25});
\draw[domain=0.29:0.31, variable=\y, c1, samples=200, smooth, ultra thick] plot (1, {\y});
\draw[domain=0.39:0.41, variable=\y, c6, samples=200, smooth, ultra thick] plot (6, {\y});
\end{tikzpicture}
\caption{Revenue v.s. $T$ for uniform distribution $U[0, 1]$}
\label{fig:revenue}
\end{figure}
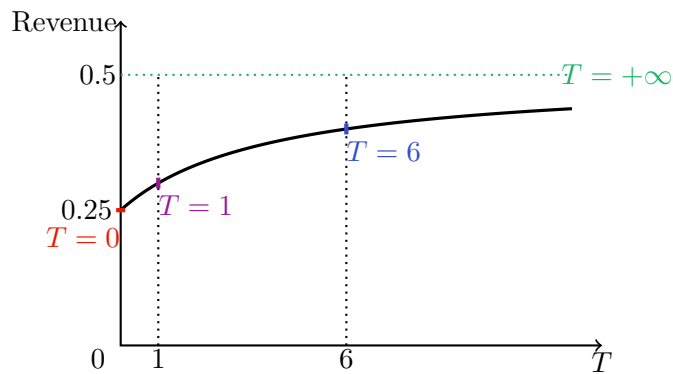

% =========== figure ===============
\cref{fig:revenue} plots the seller's optimal expected revenue in terms of the time limit $T$. Observe that if the time limit $T = 0$, the seller cannot perform price discrimination, and therefore, his best strategy is to set a fixed price of $0.5$, leading to Myerson's revenue~\citep{Myerson81}. As $T$ becomes larger and larger, the seller can increasingly exploit the impatience of the buyer and apply price discrimination to boost his revenue. In the limit case where $T = +\infty$, the seller can perfectly discriminate the buyer's values and collect the buyer's expected value as his revenue.

\section{Adaptive Lotteries}
\label{sec:adaptive}
Given that we have established the optimality of pricing mechanisms within the class of non-adaptive mechanisms (captured fully by non-adaptive sequential lottery mechanisms), a natural question is whether we can extend our result to more general classes of mechanisms, such as adaptive (sequential) lottery mechanisms. An adaptive lottery mechanism is a lottery mechanism except that the menu it posts at a timestamp can depend on the buyer's choices at previous timestamps.

In an adaptive lottery mechanism, the seller again picks a finite number $\nt$ of timestamps $t_1 \leq t_2 \leq \cdots \leq t_{\nt}$. Similar to \cref{lem:single_lottery}, without loss of generality, we assume at each timestamp, the seller provides a single lottery option of ``the item is allocated with probability $x_i$, and if the item is allocated, the payment is $p_i$''. However, the availability of any given lottery option can depend on the buyer's past choices -- if a lottery option is available, the buyer can choose it or skip; and if it is not available, the buyer has to skip (i.e., choose the null option). As explained in \cref{sec:prelim} the null option has to always be available to avoid unrealistic situations like simulating loans by the seller.
%\footnote{The skipping option has to always be available, otherwise the seller can effectively simulate loans by allocating the item upfront and charging later. The IR constraint would still be satisfied and the revenue would be unrealistically large or unbounded if the payment can be negative.}

A buyer with value $v_i$ will select a subset of lotteries: $\ell'_{i,k}$ for $k = 1, 2, \ldots, {\nt}_i$. By asking her to make this selection upfront, we get the following proposition.
\begin{proposition}
\label{prop:adaptive}
Each adaptive mechanism can be simulated with one in the following normal form: At time $t = 0$, $n$ lotteries $\ell_1, \ldots, \ell_n$ are provided, each of which has an allocation probability of $0$. Selecting each lottery $\ell_i$ disables every other $\ell_{j}$ for $j \neq i$, and exclusively opens up a sequence of later options, $\ell'_{i,k}$ for $k = 1, 2, \ldots, {\nt}_i$, each requiring the selection of all previous options.
\end{proposition}

\paragraph{Generality of Adaptive Lotteries.} Adaptive lotteries are without loss of generality. To see this, note that the most general class of mechanisms would allow the buyer to take an action from a set (or equivalently send a message) in each round, and make the allocation probability and payment when allocated in each round to be a function of the entire history. However, after applying \cref{prop:adaptive}, such a mechanism should not be any different from adaptive lotteries.

It turns out that adaptive lottery mechanisms can sometimes achieve much higher revenue than pricing mechanisms.

\begin{example}
Let $n = e^T$ where $T$ is a large variable, and let $M = 10n$. For each $i \in [n]$, $v_i = M^i$ and $f(v_i) = \frac{1}{M^i}$. The value distribution is essentially a discretized equal-revenue distribution.
\label{ex:superconst}
\end{example}
From \cref{thm:pricing_optimal}, we know the optimal non-adaptive lottery mechanism is a pricing curve. However, in the ensuing analysis, we show a huge separation between adaptive lottery mechanisms and pricing curves.
\begin{lemma}
In the example above, the revenue of any pricing curve is at most $T + 1.1$.
\end{lemma}
\begin{proof}
Let $v_i$ be the smallest value with allocation. If any $j > i$, if $t(v_{j - 1}) - t(v_j) \leq \delta$, then we can give the following upper bound on $p(v_j)$ from the IC constraint:
\[
(v_j - p(v_j)) \cdot e^{-t(v_j)} \geq (v_j - p(v_{j - 1})) \cdot e^{-t(v_{j - 1})} \geq (v_j - v_{j - 1}) \cdot e^{-t(v_{j - 1})},
\]
which leads to
\[
p(v_j) \leq v_j -  v_j \cdot \left(1 - \frac{1}{M}\right) \cdot e^{t(v_j) - t(v_{j - 1})} \leq v_j \left(1 - \left(1 - \frac{1}{M}\right)(1 - \delta)\right) \leq v_j\left(\frac{1}{M}  + \delta\right),
\]
i.e., the expected revenue from value $v_j$ is at most $v_j \cdot \left(\delta + \frac{1}{M}\right) \cdot f(v_j) = \delta + \frac{1}{M}$. Therefore, the total revenue is at most $T + \frac{n}{M} + 1$.
\end{proof}
\begin{lemma}
In the example above, the revenue of the optimal adaptive lottery mechanism is at least $0.3e^T$.
\end{lemma}
\begin{proof}
Consider the following adaptive mechanism: For each $i \in [n]$, we provide a lottery $\ell_i$ at $t = 0$. This lottery has $x_i = \frac{0.5i}{n}$ and $p_i = M^{i - 1}$. It exclusively leads to another lottery $\ell'_i$ at $t = T$, which has $x'_i = \frac{0.3}{1 - x_i}$ and $p'_i = M^i$. We show $\ell_i$ and $\ell'_i$ are the choice of a buyer with value $v_i$.

Clearly, deviating to another option of $\ell_j$ where $j > i$ provides non-positive utility. Deviating to $\ell_j$ and $\ell'_j$ where $j < i$ gives utility of
\begin{align*}
&~ (v_i - M^{j - 1}) \cdot x_j + 0.3 \cdot e^{-T} \cdot (v_i - M^j)\\
\leq &~ v_i \cdot \left(x_i - \frac{0.5}{n}\right) + 0.3 \cdot e^{-T} \cdot v_i\\
\leq &~ (v_i - M^{i - 1}) \cdot \left(x_i - \frac{0.4}{n}\right) + 0.3 \cdot e^{-T} \cdot v_i\\
\leq &~ (v_i - M^{i - 1}) \cdot x_i,
\end{align*}
which is not higher than choosing $\ell_i$ and $\ell'_i$. Thus, the revenue of this mechanism is at least $0.3n$.
\end{proof}
The example leads to the following theorem.
\begin{theorem}
\label{thm:adaptivity_gap}
Let $H = \ln \frac{v_n}{v_1}$ and $D = e^{T}$. The revenue gap between non-adaptive lottery mechanisms and adaptive ones can be $\tilde \Omega(n)$, $\tilde \Omega(H)$, and $\tilde \Omega(D)$.
\end{theorem}
Notice that the revenue gap between them is at most $O(n)$, $O(H)$, and $O(D)$ -- the revenue of optimal non-adaptive mechanism obtains both an $O(n)$-approximation and an $O(H)$-approximation to the maximum welfare even when $T = 0$; and it also obtains an $O(D)$-approximation to the revenue of adaptive ones, since every lottery in an adaptive mechanism can be moved to $t = 0$ with a discount of $\frac{1}{O(D)}$, and when $T = 0$, the optimal mechanism is to simply post a price according to Myerson's characterization. Thus \cref{thm:adaptivity_gap} gives almost tight bounds for the power of adaptivity.

Even for a value distribution of support size $3$, it is still possible that adaptive lotteries and non-adaptive ones give different revenue, illustrated in the following example.
\begin{example}
The buyer's value is drawn from $\{v_1 = 100, v_2 = 101, v_3 = 102\}$ uniformly. $T = 2 \ln 2$. The optimal pricing mechanism is to sell to $v_3$ at $t_1 = 0$ and $p_1 = 101.25$, to $v_2$ at $t_2 = \ln 2$ and $p_2 = 100.5$, and to $v_1$ at $t_3 = 2 \ln 2$ and $p_3 = 100$. (The optimal pricing mechanism can be calculated using \cref{alg:poly_time} that will be discussed in \cref{subsec:opt_alg}).

The following is an adaptive mechanism achieving more revenue: It prices at $t_1 = 0$ with $p_1 = 101.25$, and holds a lottery at $t_2 = \ln \frac{4}{3}$ with $x_2 = 0.5$ and $p_2 = 100 + \frac{1}{3}$. If the buyer picks the lottery at $t_2$, it provides a deterministic option at $t_3 = 2 \ln 2$ with $p_3 = 101$. Otherwise, it provides another deterministic option at $t_4 = 2 \ln 2$ with $p_4 = 100$.

Clearly, a buyer with value $v_1 = 100$ skips the lottery at $t_2$ and picks the option at $t_4$. For a buyer with value $v_2 = 101$, if she picks the lottery at $t_2$ then the option at $t_3$, her utility is
\[
x_2 (v_2 - p_2) e^{-t_2} + (1 - x_2) (v_2 - p_3) e^{-t_3} = \frac{1}{2} \cdot \frac{2}{3} \cdot \frac{3}{4} = \frac{1}{4}.
\]
If she skips the lottery at $t_2$ and picks the option at $t_4$, she gets utility of $\frac{1}{4}$ as well. (We can slightly perturb the prices to make the tie breaking in our favor.)

For a buyer with value $v_3 = 102$, picking the option at $t_1$ gives utility $\frac{3}{4}$. Skipping the lottery at $t_2$ and picks at $t_4$ only gives $\frac{1}{2}$. If she picks the lottery at $t_2$ then the option at $t_3$, she gets
\[
x_2 (v_3 - p_2) e^{-t_2} + (1 - x_2) (v_3 - p_3) e^{-t_3} = \frac{1}{2} \cdot \frac{5}{3} \cdot \frac{3}{4} + \frac{1}{2} \cdot 1 \cdot \frac{1}{4} = \frac{3}{4},
\]
which is not higher. Therefore, the expected payment of a buyer with value $v_1$ or $v_3$ is the same as the optimal pricing mechanism. However, the expected payment of a buyer with value $v_2$ is $100 + \frac{2}{3}$, which is strictly larger than $100.5$ obtained by the optimal pricing mechanism.
\end{example}
\begin{theorem}
\label{thm:v3diff}
There exists an instance $\langle T, V, f \rangle$  of the problem with $|V| = 3$, in which the revenue obtained by the optimal adaptive lottery mechanism is strictly higher than any pricing mechanism.
\end{theorem}

Despite these revenue gaps between adaptive lotteries and pricing curves, we can show that they differ by at most a constant multiplier when the value distribution is $\alpha$-regular~\citep{DBLP:journals/teco/ColeR17} for some $\alpha$ greater than a positive constant. $\alpha$-regular distributions are more general than monotone-hazard-rate ones, which include exponential, uniform, and normal distributions.
\begin{theorem}
\label{thm:alpha_reg}
When the value distribution is $\alpha$-regular for some $\alpha \in (0, 1)$, the multiplicative revenue gap between any adaptive lottery mechanism and the optimal pricing curve is at most $\alpha^{\frac{1}{1 - \alpha}}$. If the value distribution has monotone hazard rate (i.e. $\alpha = 1$), the gap is at most $e$.
\end{theorem}

We first prove the following lemma.
\begin{lemma}
\label{lem:revenuewelfarebound}
The revenue of an adaptive lottery mechanism is at most the maximum welfare, i.e. the expected value of the buyer $\sum_{i = 1}^n v_i f(v_i)$.
\end{lemma}
This statement is not immediate: a mechanism can sell a lottery at a price higher than the value of the buyer, promising a good deal will be provided later.
\begin{proof}
Fix value $v_i$ and we will drop the subscript $i$ from now on. The goal is to show the undiscounted payment of a buyer with value $v$ is at most $v \cdot X$, where $X$ is the allocation probability if the value is $v$. We use induction on $\nt$, the number of lotteries for $v$. The base case where $\nt = 1$ is immediate. For $\nt = k$, as the bidder's utility is non-negative, we have:
\[
\sum_{j = 1}^k X_j (v - p_j) e^{-t_j} \geq 0,
\]
where $X_j$ is the probability that all lotteries before $\ell_j$ failed and $\ell_j$ succeeded, $p_j$ is the unit price of $\ell_j$ and $t_j$ is the time of $\ell_j$. This is equivalent to
\[
\sum_{j = 1}^k \sum_{s = 1}^j X_s (v - p_s) (e^{-t_j} - e^{-t_{j + 1}}) \geq 0,
\]
where $t_{k + 1} = +\infty$ for simplicity.
Thus, there is some $j \in [k]$, so that $\sum_{s = 1}^j X_s (v - p_s) (e^{-t_j} - e^{-t_{j + 1}}) \geq 0$, i.e., $\sum_{s = 1}^j X_s (v - p_s) \geq 0$. Combining with the inductive hypothesis that $\sum_{s = j + 1}^k X_s (v - p_s) \geq 0$, we get the desired $\sum_{s = 1}^k X_s (v - p_s) \geq 0$.
\end{proof}
\begin{proof}[Proof of \cref{thm:alpha_reg}]
By \cref{lem:revenuewelfarebound}, the revenue of any adaptive lottery mechanism is at most the maximum welfare, which is at most $\alpha^{\frac{1}{1 - \alpha}}$ times Myerson's revenue on that value distribution~\citep{DBLP:journals/teco/ColeR17}. A pricing curve can achieve at least Myerson's revenue by posting Myerson's price throughout the horizon.
\end{proof}
In general, whenever the revenue-to-welfare ratio of the value distribution in the static pricing problem is not very small, the power of adaptivity in our problem is limited.

\section{Continuous distributions}
\label{sec:continuous}
In this section, we show that the optimal pricing curve for continuous %\footnote{Actually, the result in this section holds for general value distributions. We use the name ``continuous'' to avoid confusion.}
distributions can be approximated by first discretizing the distribution and then computing the optimal pricing curve on the discretized distribution. For notational convenience, we use $\rev(P,D)$ to denote the revenue of pricing curve $P$ on value distribution $D$ for some fixed discount multiplier.

\begin{definition}
\label{def:disc}
Let $D$ be a distribution on $[0, M]$. Let $F_D$ be the CDF of $D$. We discretize $D$ on the quantile space with an integer parameter $k > 0$: 
\begin{itemize}
    \item $D^-$ is a uniform distribution over $k$ values: $F^{-1}_D(0), F^{-1}_D\left(\frac{1}{k}\right), \ldots, F^{-1}_D\left(\frac{k - 1}{k}\right)$.
    \item $D^+$ is a uniform distribution over $k$ values: $F^{-1}_D\left(\frac{1}{k}\right), \ldots, F^{-1}_D\left(\frac{k - 1}{k}\right), F^{-1}_D(1)$.
\end{itemize}
\end{definition}

\begin{lemma}
\label{lem:dominance}
$D^+$ (first-order) stochastically dominates $D$, and $D$ stochastically dominates $D^-$.
\end{lemma}
\begin{proof}
For any $F^{-1}_D(0) \leq x < F^{-1}_D(1)$, there is some $i \in \{0, \ldots,k-1\}$ so that $F^{-1}_D\left(\frac{i}{k}\right) \leq x < F^{-1}_D\left(\frac{i + 1}{k}\right)$. We know $\Pr_{X \sim D} [X \leq x] \in \big[\frac{i}{k}, \frac{i + 1}{k}\big)$, and by construction $\Pr_{X \sim D^-}[X \leq x] = \frac{i + 1}{k}$ and $\Pr_{X \sim D^+}[X \leq x] = \frac{i}{k}$. Thus, $\Pr_{X \sim D^+}[X \leq x] \leq \Pr_{X \sim D} [ X \leq x] \leq \Pr_{X \sim D^-}[X \leq x]$ for  $x \in [F^{-1}_D(0), F^{-1}_D(1))$. Similar arguments would show this is also true for $x = F^{-1}_D(1)$. Therefore we get the stochastic dominance.
\end{proof}

\begin{lemma}
\label{lem:sandwich}
For any pricing curve $P$, $\rev(P, D^+) \leq \rev(P, D^-) + \frac{M}{k}$.
\end{lemma}
\begin{proof}
Notice that $D^-$ and $D^+$ only differ in $\frac{1}{k}$ fraction of the distribution and the revenue difference from that would be at most the max value $M$. Therefore, $\rev(P, D^+) \leq \rev(P, D^-) + \frac{M}{k}$.
\end{proof}

\begin{lemma}
\label{lem:sto_mono}
Consider two value distributions $D, D'$ where $D$ stochastically dominates $D'$. Then, for any pricing curve $P$, $\rev(P,D) \geq \rev(P,D')$.
\end{lemma}
\begin{proof}
By the definition of stochastic dominance, it suffices to show that for any values $v > v'$, $v$ pays no less than $v'$ on the pricing curve $P$. Let $v$ get discounted allocation $x(v) = e^{-t(v)}$ and unit price $p(v)$ and $v'$ gets discounted allocation $x(v') = e^{-t(v')}$ and unit price $p(v')$ on pricing curve $P$. By the optimality of $v$ and $v'$'s choices, we have
\[
x(v) \cdot (v - p(v)) \geq x(v') \cdot (v - p(v')),
\]
and 
\[
x(v') \cdot (v' - p(v')) \geq x(v) \cdot (v' - p(v)).
\]
Summing them up, we get 
\[
(x(v) - x(v')) \cdot (v - v') \geq 0.
\]
Therefore $x(v) \geq x(v')$. Plugging this back into the inequality above, we get
\[
x(v) (v' - p(v)) \leq x(v') (v' - p(v')) \leq x(v) (v' - p(v')).
\]
Therefore, $x(v) p(v) \geq x(v) p(v')$. If $x(v) > 0$, we get $p(v) \leq p(v')$. If $x = 0$, since $x(v') \leq x(v)$, we know $x(v') = 0$ and then $p(v) = p(v') = 0$. Thus, in both cases, we have $p(v) \geq p(v')$ and this finishes the proof.
\end{proof}

\begin{theorem}
\label{thm:disc}
Let $D$ be a distribution on $[0, M]$. Let $D^+$ and $D^-$ be the quantile discretization of $D$ in \cref{def:disc} with parameter $k$. Let $P^+$, $P^-$, $P^*$ be the optimal pricing curves of $D^+$, $D^-$, $D$. We have that the optimal revenue calculated using the discretized distribution approximates the the optimal revenue:
\[
\big|\rev(P^+, D^+) - \rev(P^*,D)\big| \leq \frac{M}{k},
\]
and the pricing curve calculated using the discretized distribution has approximately optimal revenue on the actual distribution:
\[
\big|\rev(P^+, D) - \rev(P^*,D)\big| \leq \frac{M}{k}.
\]
\end{theorem}
\begin{proof}
By \cref{lem:dominance}, we know $D^+$ stochastically dominates $D$, and $D$ stochastically dominates $D^-$. Therefore, by \cref{lem:sto_mono} and the optimality of $P^+, P^-, P^*$, we have
\[
\rev(P^+, D^+) \geq \rev(P^*, D^+) \geq \rev(P^*, D) \geq 
\rev(P^-, D) \geq \rev(P^-, D^-)\geq \rev(P^+, D^-). 
\]
We also have
\[
\rev(P^+, D^+) \geq \rev(P^+,D) \geq  \rev(P^+,D^-).
\]
Therefore, both of $\rev(P^*, D)$ and $\rev(P^+,D)$ are within the range $[\rev(P^+, D^-), \rev(P^+, D^+)]$. On the other hand, by \cref{lem:sandwich}, we have
\[
\rev(P^+, D^-) \geq \rev(P^+,D^+) - \frac{M}{k},
\]
which concludes the proof.
\end{proof}
\cref{thm:disc} enables us to discretize the value distribution and compute the optimal pricing curve on the discretized version. As $k \to \infty$ (the discretization becomes finer), we know the revenue difference goes to $0$.

\bibliographystyle{plainnat}
\bibliography{ref}

%\appendix

\end{document}